\documentclass[a4paper,twocolumn,11pt,accepted=2025-01-24]{quantumarticle}
\pdfoutput=1
\usepackage[utf8]{inputenc}
\usepackage[english]{babel}
\usepackage[T1]{fontenc}
\usepackage{amsmath}
\usepackage{hyperref}
\usepackage[numbers]{natbib}
\usepackage{listings}

\usepackage[resetlabels]{multibib}

\newcites{article}{article references}
\newcites{book}{book references}
\newcites{misc}{misc references}
\newcites{repo}{repository references}
\newcites{web}{website references}
\newcites{other}{Other references}

\usepackage{tikz}
\usepackage{lipsum}

\usepackage{amsthm}
\usepackage{amssymb}
\usepackage[version=4]{mhchem}

\newtheorem{proposition}{Proposition}

\usepackage{graphicx}

\graphicspath{{./Images/}}
\usepackage{xcolor}
\usepackage{verbatim}
\usepackage{mathtools}
\usepackage{bm}
\usepackage{yhmath}
\usepackage[inline]{asymptote}
\usepackage{cancel}
\usepackage{relsize}
\usepackage{array}
\usepackage{bm}
\usepackage{glossaries}
\usepackage{cleveref}

\newacronym{onv}{ONV}{occupation number vector}
\newacronym{qpe}{QPE}{quantum phase estimation}
\newacronym{qft}{QFT}{quantum Fourier transform}
\newacronym{cnt}{CNT}{carbon nanotube}
\newacronym{fccvd}{FCCVD}{floating catalytic chemical vapor deposition}
\newacronym{cf}{CF}{carbon fiber}

\newcommand{\set}[1]{\left\{ #1 \right\}}

\newcommand{\group}[1]{\left( #1 \right)}
\newcommand{\norm}[1]{\left\|#1\right\|}

\newcommand{\op}[1]{\mathcal{#1}} %consistent notation for general "operator"
\newcommand{\todo}[1] {\textcolor{red}{ToDo: #1}}

\newcommand{\bra}[1]{\left\langle #1 \right|}
\newcommand{\ket}[1]{\left| #1 \right\rangle}

\newcommand{\ord}{{\mathcal O}}

\newcommand{\bea}{\begin{eqnarray}}
\newcommand{\ea}{\end{eqnarray}}
\renewcommand{\emptyset}{\varnothing}
\newcommand{\ceiling}[1]{\left\lceil #1 \right\rceil}
\newcommand{\defn}{:=}

\begin{document}

%\preprint{APS/123-QED}

\title{Refining resource estimation for the quantum computation of vibrational molecular spectra through Trotter error analysis}

\author{Dimitar Trenev}
\affiliation{ExxonMobil Technology and Engineering Company, Annandale, NJ 08801, USA}
\author{Pauline J Ollitrault}
\affiliation{IBM Quantum, IBM Research Zurich, S{\"a}umerstrasse 4, 8803 R{\"u}schlikon, Switzerland}
\altaffiliation[Present address: ]{QC Ware, Palo Alto, CA, USA}
\author{Stuart M. Harwood}
\affiliation{ExxonMobil Technology and Engineering Company, Annandale, NJ 08801, USA}
\author{Tanvi P. Gujarati}
\affiliation{IBM Quantum, IBM Research Almaden, San Jose, CA 95120, USA}
\author{Sumathy Raman}
\affiliation{ExxonMobil Technology and Engineering Company, Annandale, NJ 08801, USA}
\author{Antonio Mezzacapo}
%\affiliation{IBM Quantum, IBM T.J. Watson Research Center, Yorktown Heights, NY 10598, USA}
\author{Sarah Mostame}
\email{sarah.mostame@ibm.com}
\affiliation{IBM Quantum, IBM T.J. Watson Research Center, Yorktown Heights, NY 10598, USA}
\maketitle
\begin{abstract}
Accurate simulations of vibrational molecular spectra are expensive on conventional computers. Compared to the electronic structure problem, the vibrational structure problem with quantum computers is less investigated. In this work we accurately estimate quantum resources, such as number of logical qubits and quantum gates, required for vibrational structure calculations on a programmable quantum computer. Our approach is based on quantum phase estimation and focuses on fault-tolerant quantum devices. In addition to asymptotic estimates for generic chemical compounds, we present a more detailed analysis of the quantum resources needed for the simulation of the Hamiltonian arising in the vibrational structure calculation of acetylene-like polyynes of interest. Leveraging nested commutators, we provide an in-depth quantitative analysis of trotter errors compared to the prior investigations. Ultimately, this work serves as a guide for analyzing the potential quantum advantage within vibrational structure simulations.
\end{abstract}

%\tableofcontents
%%%%%%%%%%%%%%%%%%%%%%%%%%%%%%%%%%%%%%%%%%%%%%%%%%%%
%%%%%%%%%%%%%%%%%%%%%%%%%%%%%%%%%%%%%%%%%%%%%%%%%%%%

\section{Introduction}
Attaining the electronic structure of molecules is one of the core problems  in quantum chemistry and material design \cite{Marzari2021}. Although the electronic structure of small molecules can be obtained with good accuracy, solving the electronic structure of larger molecular systems continues to be challenging and computationally intensive. 
A vast number of methods -- such as Hartree-Fock, configuration interaction, coupled cluster theory, and density functional theory -- have been developed to address the accuracy and efficiency in treating the molecular electronic-structure problems \cite{MESTheory,Leach_CC,Cramer_CC,Jensen_CC}. 
Quantum computing
%, on the other hand, 
offers an alternative approach\footnote{beyond these quantum mechanical approaches on classical computers} to tackle the computational challenges associated with these problems using quantum hardware \cite{doi:10.1126/science.1113479,Liu2022,Cao2019}.
Significant advances have been made in the past two decades in the development of quantum algorithms to harness the capabilities of currently available noisy quantum devices, particularly for addressing electronic-structure problems~\cite{Lanyon2010,Kandala2017,Peruzzo2014,Cerezo2021}.
Additionally, researchers have estimated the quantum resources required for simulating molecular electronic structure on future fault-tolerant quantum devices ~\cite{reiher2017}.

In the realm of modern chemistry, electronic structure calculations play a crucial role in obtaining structural properties and spectroscopic characteristics of molecules, as descriptors enabling screening of materials\footnote{e.g. active materials such as catalysts, sorbents and membranes for separation} in industrial applications.
However, the quest for accurate and efficient electronic structure methods is not the sole computational obstacle in quantum chemistry and material science. The vibrational structure problem represents another fundamental challenge.
To make a significant impact in both fundamental science and industrial applications, it is necessary to go beyond the electronic structure and develop a model that requires a comprehensive understanding of the molecular vibrational structure, e.g., to determine reaction rate constants~\cite{Laidler, Sumathi2002}. 
While classical computers can handle the electronic structure problem of small molecules with reasonable accuracy, the same does not hold true for calculating the vibrational structure of molecules beyond the harmonic approximation.
Previous works suggest that quantum-computing approaches for calculating vibrational structure have the potential to reach quantum advantage prior to their electronic-structure counterparts, for the respective required energy precision~\cite{PhysRevA.104.062419}. 
Resource estimate for a variety of bosonic operators has been provided in \cite{sawayanpj2020}, employing different encoding methods.
Nevertheless, the quantum computing applications for molecular vibrational structure calculations have received limited scholarly attention. 
Only a modest number of studies have been conducted to date, and these studies  marginally address the quantum resource requirements~\cite{sawayaJPCL2019,ollitraultCS2020,sparrowNature2018,magannPRR2021,majland2022,jahangiriPCCP2020,richerme2023,dalzell2023,D3SC02453A}.
The authors in~\cite{PhysRevA.104.062419} present algorithms for calculating vibrational spectra on both near- and long-term quantum devices. They compare concrete quantum resources such as qubit count, as well as proxy quantities determining computational complexity, e.g. Hamiltonian magnitude, for both electronic and vibrational  structure simulations. 
Lastly, they present numerical error analyses on several
vibrational Hamiltonian examples - carbon monoxide, the isoformyl radical, ozone, and a model Hamiltonian of Fermi resonance - as a proof of principle. 
They speculate quantum advantage will take place for a real-world vibrational structure problem instance before any electronic structure one. 
However, they suggest further investigation on correlation between vibrational problem instances and Trotter error.

Here, we present a more extensive study of the Trotter error for vibrational structure simulations based on the recent theory in ref.~\cite{childs2021theory}.
Our focus in this work is on providing quantum resource  estimation for  simulating molecular vibrational structure.  
It should be noted that to present results independent of particular hardware and error correction strategies, and given the considerable research and progress in these areas — such as recent advancements in low-density parity-check (LDPC)~\cite{LDPC2024} — we focus on the logical quantum resource requirements without considering any specific error-correction method.
To improve the previous studies, we start with a distinct Hamiltonian model to go beyond harmonic approximation and obtain precise Hamiltonians for the examples represented in this work. 
Our results are based on the ``$L$-mode representation" of the Hamiltonian and the vibrational self-consistent field (VSCF) method~\cite{hansen2010new}.
We first offer an asymptotic analysis and then provide a more quantitative study on industrially-relevant chemical compounds.
More specifically, our examples focus on polyyne molecules, which are key intermediates in the commercial conversion of light gases into valuable materials such as carbon nanotubes and carbon fibers.

%More specifically, our focus in the examples here is on polyyne molecules, which are the dominant intermediates\footnote{as suggested by molecular dynamics simulations} in the commercial upgrade of light gases to the more attractive carbon nano-tubes and carbon-fiber materials.

In the next section we will introduce a quantum approach for vibrational structure calculations and analyze its computational complexity. As mentioned above, our study here is based on the $L$-mode Hamiltonian representation which is more suitable\footnote{as opposed to, e.g., quartic force field~\cite{Lee1995}.} for asymptotic analysis. We review mapping of the vibrational Hamiltonian to qubits and present comparison of qubit requirements for the “unary” and “binary” encoding methods.  Note that throughout this article the term “qubit” refers to “logical qubits” and we will not discuss any specific error-correction methods. 
We continue \Cref{sec:QC_algo} by estimating the complexity of the algorithm (based on quantum phase estimation) when Trotterization is used for approximate Hamiltonian evolution. 
We continue \Cref{sec:QC_algo} by estimating the complexity of the algorithm (based on quantum phase estimation) when Trotterization is used for approximate Hamiltonian evolution. 
In \Cref{sec:Improving_Trott}, we focus on improving the accuracy of our asymptotic studies, following up on the results from~\citet{childs2021theory}. In \Cref{sec:RE_polyynes}, we deliver more precise numerical results on the quantum computational cost for a set of industrially-relevant chemical compounds, namely acetylene-like polyyne molecules.
Finally, we provide a summary of our results, and discuss the themes for future research in comparing the quantum resource requirements for vibrational and electronic structure problems. 

%%%%%%%%%%%%%%%%%%%%%%%%%%%%%%%%%%%%%%%%%%%%%%%%%%%%
%%%%%%%%%%%%%%%%%%%%%%%%%%%%%%%%%%%%%%%%%%%%%%%%%%%%

\section{Quantum algorithm for vibrational structure calculations}
\label{sec:QC_algo}
\begin{figure}[ht]
    \centering
    \includegraphics[width = 0.48\textwidth]{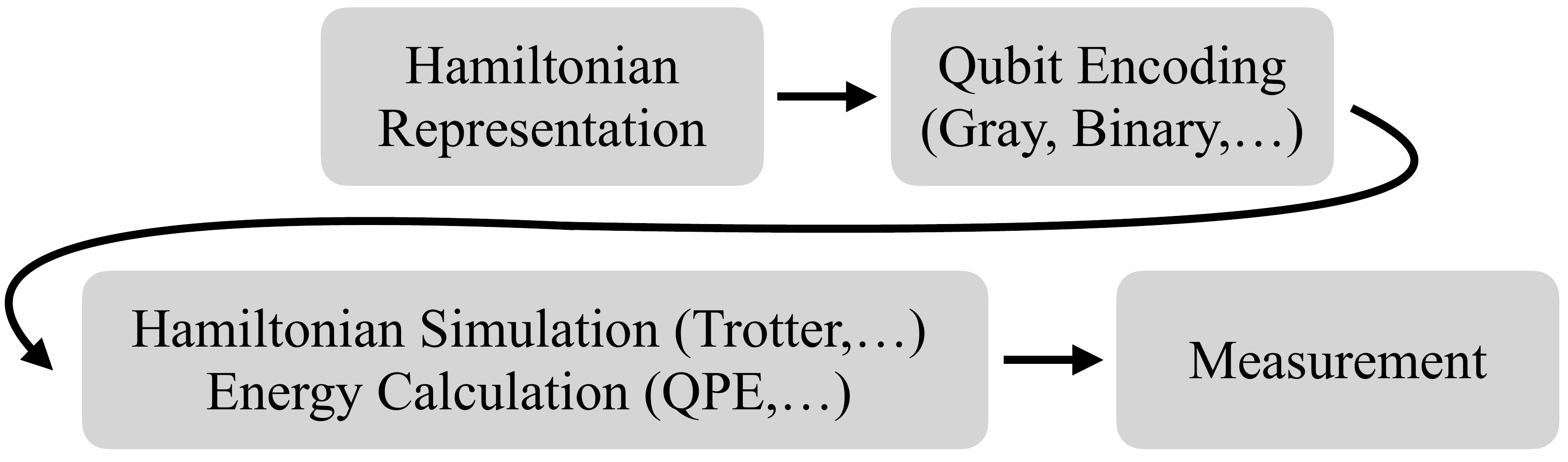}
    \label{fig:diagram}
\end{figure}

Like other quantum algorithms, as illustrated in the diagram above, a quantum algorithm for vibrational structure problems is consisted of several distinct steps. In the following subsections, we will delve into these steps in detail, tying everything together to provide asymptotic resource estimates in \cref{sec:Ham_evol}. 
We should note that a pivotal aspect of Hamiltonian simulation involves determining the parameters, such as the number of Trotter steps, to yield results sufficiently accurate for the application at hand. This may necessitate some preprocessing of the qubit Hamiltonian terms (falling between the second and third steps of the diagram). We will discuss this further in \cref{sec:Improving_Trott}. 

%%%%%%%%%%%%%%%%%%%%%%%%%%%%%%%%%%%%%%%%%%%%%%%%%%%%

\subsection{The vibrational Hamiltonian}

A general expression for the vibrational Hamiltonian is given by
the multi-mode expansion, called $L$-mode representation, which expands the potential energy surface into a sum of
one-mode potentials, two-mode potentials, three-mode potentials, and so on~\cite{Bowman2003,Li2001}. The Hamiltonian for a system with $L$ modes of vibration each denoted by a variable $Q_{.}$ is given as:
\begin{align}
\mathcal{H}_{\textrm vib}(Q_1, \ldots, Q_L) &=
    - \frac{1}{2} \sum_{l=1}^{L} \frac{\partial^2}{\partial Q_l^2}
    + V(Q_1, \ldots, Q_L)
    \nonumber 
\end{align}
\begin{align}
V(Q_1, \ldots, Q_L) &=
    V_0 + \sum_{l=1}^L V^{[l]}(Q_l) 
    \nonumber \\
    & + \sum_{l<m}^L V^{[l,m]}(Q_l, Q_m) 
    \nonumber \\
    & + \sum_{l<m<n}^L V^{[l,m,n]}(Q_l, Q_m, Q_n) + \dots
\end{align}

\begin{table}[t]
    \caption{The variables and their corresponding meaning}
    \label{tab:Variables naming}
    \centering
    \resizebox{\columnwidth}{!}{%
    \begin{tabular}{cl}
        \hline
        Symbol & Meaning \\
        \hline
         $D$ & Truncation order of the Hamiltonian expansion  \\
         $d$ & Truncation order of the bosonic modes occupation \\
         & (i.e. number of modals)\\
         $L$ & Number of modes \\
         $p$ & Order of the Trotter expansion \\    
    \hline
    \end{tabular}%
    }
\end{table}

The Hamiltonian expansion is truncated at arbitrary order $D$ which can take integer values up to $L$.
%so that if $D=1$, the truncation happens after the $V^{[l]}(Q_l)$ terms; if $D=2$, the truncation happens after the $V^{[l,m]}(Q_l, Q_m)$ terms etc. 
%%%
The occupation number of each bosonic mode (i.e., the number of modals per mode $d_l$) can vary, and to determine the exact resource requirement, one needs to consider different modals per mode. 
From a chemistry perspective, the distribution of modals per mode depends on the energies of the modes. 
Given that the focus of our work is on estimating an upper bound for quantum resources, we opt to take \mbox{$d_l={\rm{max}}_i \, d_i = d$}, a number of modals that is sufficient to achieve a given chemical accuracy~\cite{hansen2010new, Hansen2008, Heislbetz2010}. This allows us to estimate resource requirements conservatively, ensuring the upper bound is robust.
When specific values of $d_l$ are provided, calculating a more refined resource estimate becomes straightforward. For example, the number of required qubits will be  $\sum_l d_l $ instead of $Ld$.
%
%We also assume that the maximum number of modals per mode $d_{\rm max$ (i.e., occupation number of each bosonic mode) is truncated to order $d=d_{\rm max$.
%
The second quantization is introduced through a general set of one-particle basis functions called modals. To each mode, $l$, is then associated a number $d$ of modals $\{\phi_{n_l}\}$. 
The Hamiltonian becomes
\begin{align}
\label{eq:Hamil2ndquan}
&\mathcal{H}_{\textrm vib}^{SQ} =
    \sum_{l=1}^L \sum_{k,h}^{d} \bra{\phi_{k_l}}T(Q_l) + V^{[l]}(Q_l) \ket{ \phi_{h_l} } a_{k_l}^\dagger a_{h_l} \nonumber \\    
    &+ \sum_{l<m}^L \sum_{k,h}^{d} \sum_{k,h}^{d}
    \Big\{\bra{\phi_{k_l} \phi_{k_m}}  V^{[l,m]}(Q_l, Q_m) \ket{\phi_{h_l} \phi_{h_m}} 
    \nonumber \\
     & \left(a_{k_l}^\dagger a_{k_m}^\dagger a_{h_l} a_{h_m} \right) \Big\} + \dots
\end{align}
where $a_{\cdot}^{\dagger}$ and $a_{\cdot}$ are the creation and annihilation operators, respectively, discussed below.
The number of terms in the Hamiltonian then amounts to 
\begin{equation}
    \label{eq:num_terms_asymp}
    N_H = \ord\left((L d^2)^{D}\right)\,.
\end{equation}
The various truncation orders and other fundamental sizing parameters are summarized in \cref{tab:Variables naming}.

Note that in practice, there are two approaches commonly employed to obtain modals and construct the system Hamiltonian: the VSCF method and the Harmonic approximation.
In this work, we use VSCF method based on second quantization~\cite{hansen2010new,csaszar2012} and our focus will be on the anharmonic vibrational wave function, see Section I in the \textit{Supplementary Information} for more details.

%%%%%%%%%%%%%%%%%%%%%%%%%%%%%%%%%%%%%%%%%%%%%%%%%%%%

\subsection{Qubit encodings}
Quantum operators in the Hamiltonian~\eqref{eq:Hamil2ndquan} act on indistinguishable bosons, which need to be mapped onto distinguishable qubits. 
Qubit encoding facilitates this mapping by transforming the  bosonic Fock space into the qubit Hilbert space, where each bosonic state is represented by a corresponding qubit state.
In this section, we provide a comparison between \textit{unary} and \textit{binary} encodings, see ref.~\cite{sawayanpj2020} for example. Specifically, we focus our discussion on the standard binary encoding. It is important to note that there is an exponential number of ways to encode $d$ integers in $\log_2(d)$ bits allowing for the existence of alternative binary encoding methods. One such example is the Gray code, which ensures the binary representations of neighboring integers have a Hamming distance of one. In practice, the choice of mapping can be customized and optimized according to the specific system of interest. 

With the unary encoding, the many-body basis functions $\phi_{k_1} \cdots \phi_{k_L}$ can be encoded as an \gls{onv} as
\begin{align}
\label{eq:ONV_Christiansen}
  \phi_{k_1} \cdots \phi_{k_L}
                      &\defn
                      \nonumber \\
    &\ket{0_1 \cdots 1_{k_1} \cdots 0_{d},
                                   \cdots , 
                                   0_1 \cdots 1_{k_L} \cdots 0_{d}}.
\end{align}
Based on this representation, the creation and annihilation operators act on modal $k_l$ of mode $l$ as:
\begin{align}
  \label{eq:SQBosonicDefinition}
    a_{k_l}^\dagger \ket{ \cdots, 0_1 \cdots 0_{k_l} \cdots 0_{d}, \cdots} &=  
    \nonumber \\
    &\ket{ \cdots, 0_1 \cdots 1_{k_l} \cdots 0_{d}, \cdots} \nonumber \\
    a_{k_l}^\dagger \ket{ \cdots, 0_1 \cdots 1_{k_l} \cdots 0_{d}, \cdots} &= 0 
    \nonumber \\
    a_{k_l} \ket{ \cdots, 0_1 \cdots 1_{k_l} \cdots 0_{d}, \cdots} & = 
    \nonumber \\
    &\ket{ \cdots, 0_1 \cdots 0_{k_l} \cdots 0_{d}, \cdots} 
    \nonumber \\
    a_{k_l} \ket{ \cdots, 0_1 \cdots 0_{k_l} \cdots 0_{d}, \cdots} &= 0 
\end{align}
with
\begin{equation}
  \begin{aligned}
    \left[ a_{k_l}^\dagger, a_{h_m}^\dagger \right] &= 0\, ,  \quad
    \left[ a_{k_l}, a_{h_m} \right] &= 0 \, , \, \,   \nonumber \\
     \quad
    \left[ a_{k_l}, a_{h_m}^\dagger \right] &= \delta_{l,m} \, , \delta_{k_l,h_m} \,.
  \end{aligned}
  \label{eq:BosonicCommutationRule}
\end{equation}
Therefore the mapping to Pauli operators is straightforward:
\begin{equation}
    a^{\dagger}_{k_l} = \sigma^x_{k_l} - i \sigma^y_{k_l} \, , \quad \text{and} \qquad
    a_{k_l} = \sigma^x_{k_l} + i \sigma^y_{k_l} \, .
\end{equation}
%and
%\begin{equation}
%    a_{k_l} = \sigma^x_{k_l} + i \sigma^y_{k_l}.
%\end{equation}

In the case of a binary mapping, the situation is different. 
Any operator acting on a $d$-level system can be written as:
\begin{equation}
    \op{A} = \sum_{l,l'}c_{l,l'}| l \rangle\langle l'| \, , 
    \label{eq:bin_2q_op}
\end{equation}
where $\ket{l}$ is a $\log_2(d)$-long bitstring encoding the binary representation of the integer $l$. Each term can be rewritten in terms of the tensor product of individual qubit states $\ket{l}\bra{l'} = \bigotimes_j \ket{x_j}\bra{x_j'}$.
To map the bosonic operators $a^{\dagger}_{k_l}$ and $a_{k_l}$ to Pauli operators, one first has to translate the products $a^{\dagger}_{k_l}a_{h_l}$, $a^{\dagger}_{k_l}a^{\dagger}_{k_m}a_{h_l}a_{h_m}$, etc, to operators of the form of Eq.~\eqref{eq:bin_2q_op}.
For instance, consider now the \gls{onv} to be $\ket{k_1, \cdots, k_L}$ with each $\ket{k_l}$ the binary representation of integer $k_l$. Then a one-body product can be written as 
\begin{equation}
    a^{\dagger}_{k_l}a_{h_l} = \ket{\cdots, k_l, \cdots}\bra{\cdots, h_l, \cdots}\, .
\end{equation}
Then the following formulae can be employed,
\begin{equation}
  \begin{aligned}
    &\ket{0}\bra{1} = \frac{1}{2} (\sigma^x + i\sigma^y) \, ,  \quad  
    \ket{1}\bra{0} = \frac{1}{2} (\sigma^x - i\sigma^y)\, , \\
    &\ket{0}\bra{0} = \frac{1}{2} (I + \sigma^z)
    \, ,  \qquad 
    \ket{1}\bra{1} = \frac{1}{2} (I - \sigma^z) \, .\\
  \end{aligned}
\end{equation}

The unary encoding requires a total number of qubits \mbox{$N_q = Ld$} whereas any binary encoding leads to \mbox{$N_q = L\log_2(d)$}.
After mapping, the number of Pauli terms in the unary Hamiltonian is given by the number of terms in the second quantized Hamiltonian only, as the unary mapping leads to a constant number of Paulis per term. 
On the other hand, the binary encoding leads to a total number of Pauli terms scaling as \mbox{$\ord(N_H\log_2(d))$}.
For the unary encoding the Pauli terms are $2D$-local, meaning that, at maximum order, they act non-trivially on $2D$ qubits at a time.
In contrast, the binary encoding leads to operators acting on \mbox{$\ord(2 D \log_2(d))$} qubits. 
Note that this is an upper bound as in principle the encoding can be optimized so that terms involving strings with large Hamming distance are associated to negligible Hamiltonian coefficients. The choice of the number of modals $d$ is application specific and is generally determined by the temperature at which the property of interest is studied. For the polyyne molecules that we will discuss in \Cref{sec:RE_polyynes}, $d$ is small and hence unary and binary mapping for total number of Pauli terms and locality of Pauli terms are comparable. 
As we will explore further, in most applications, $D=3$ proves to be a satisfactory choice for vibrational structure problems, resulting in highly local Pauli terms in the Hamiltonian in comparison to the fermionic case. 
Indeed, mapping fermionic to qubit Hamiltonians in electronic structure simulations requires tracking the parity of the occupation numbers of orbitals in order to satisfy anti-commutation relations of the fermionic operators. This leads to creation of Pauli operators that have a locality of $\ord(N)$ where $N$ is the number of qubits required to represent a $N$ spin-orbital system~\cite{Bravyi2002} for e.g. Jordan-Wigner mapping.
While there exist other fermion to qubit mappings that are \mbox{$\ord(\log(N))$} local like the Bravyi-Kitaev mapping, we focus on the Jordan-Wigner mapping for comparison in this discussion as it is more widely used than the other mappings. 

%%%%%%%%%%%%%%%%%%%%%%%%%%%%%%%%%%%%%%%%%%%%%%%%%%%%

\subsection{Quantum Phase Estimation}
\begin{figure}[t]
    \centering
    \includegraphics[width = 0.49\textwidth]{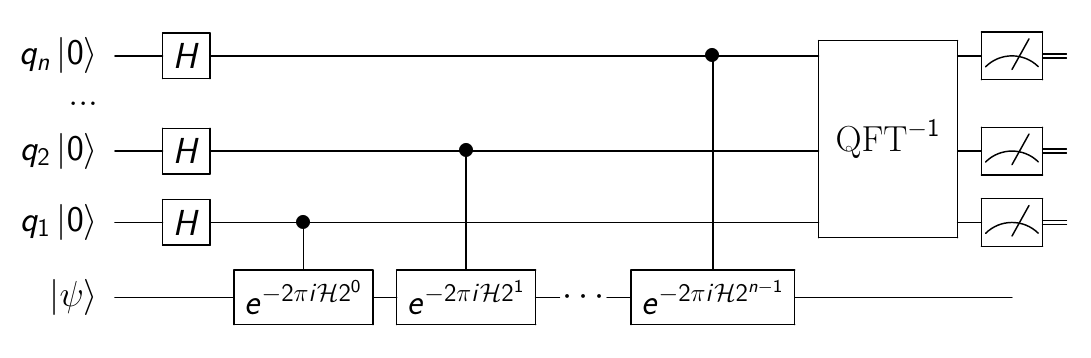}
    \caption{The canonical implementation of the \gls{qpe} algorithm in a quantum circuit}
    \label{fig:qpe}
\end{figure}

Eigenvalues of a given Hamiltonian, $\op{H}$, can be evaluated with the quantum phase estimation (\gls{qpe}) algorithm. The canonical implementation of this algorithm in a quantum circuit is given in \Cref{fig:qpe}. 
First an eigenstate\footnote{In practice, the exact eigenstate is usually not known and an approximation is used instead, slightly degrading the probability of success of the circuit. This is then remedied with multilpe circuit evaluations in the algorithm.} , $\ket{\psi}$, corresponding to eigenvalue $\lambda$, is prepared in the quantum register. 
An ancilla register of $n$
%$\omega - 1$
qubits is initialized in full superposition state. 
Then  $\ket{\psi}$ is evolved under powers of $\exp(-2\pi i \op{H})$. 
The evolution operators are controlled by the ancilla qubits as shown in \Cref{fig:qpe}. At this point the state in the quantum register is
\begin{equation}
    \ket{\psi'} = \frac{1}{2^{\frac{n}{2}}}\sum_{x=0}^{2^n} e^{-2\pi i \lambda x}\ket{x}.
\end{equation}
Assuming $\lambda$ belongs to [0,1], it can be approximated by $\gamma = y/2^n$, with $y$ being an integer. 
The inverse \gls{qft} then kicks back the phase to the state of the ancilla register. 
If $\lambda = \gamma$, the inverse \gls{qft} leaves the ancilla register in the product state $\ket{y}$ which is then read with $100\%$ probability of success. 
However, in the more general case where $|\lambda - \gamma| > 0$, $\lambda$ can be obtained with $k$ correct digits and probability of at least 0.75 if the number of ancilla qubits is $n = k + \log_2(4) = k+2$~\cite[Eqn.~(5.35)]{mikeandike}.

To enforce the assumption that the eigenvalue is in $[0,1]$, \gls{qpe} is applied to a shifted and scaled Hamiltonian $\tilde{\op{H}}$ with spectrum contained in $[0,1]$:
$\tilde{\op{H}} = \frac{1}{2\beta}(\op{H} + \beta I)$, where $[-\beta, \beta]$ contains the spectrum of $\op{H}$.
Usually, $\beta$ is simply taken to be the norm of the Hamiltonian $\beta = \norm{\op{H}}$.
In particular, when the Hamiltonian is a weighted sum of Pauli operators such that \mbox{$\op{H} = \sum_i c_i \op{P}_i$}, then \mbox{$\norm{\op{H}} = \sum_i |c_i|$}. 
We call $\epsilon_{\nu}$ the targeted spectroscopic accuracy of our calculations.
In practice, for vibrational structure problems, $\epsilon_{\nu}$ is often set to $1~\text{cm}^{-1}$.
If \gls{qpe} evaluates the eigenvalue of the scaled Hamiltonian with precision $2^{-k}$, then we set $k = \ceiling{\log_2(2\beta\epsilon_{\nu}^{-1})}$ to achieve the desired accuracy.
Combining these requirements, the number of ancilla qubits required in the \gls{qpe} to read the eigenvalue with precision $\epsilon_{\nu}$ and probability of at least 0.75 is given by $n = \ceiling{\log_2(8\beta\epsilon_{\nu}^{-1})}$.

%%%%%%%%%%%%%%%%%%%%%%%%%%%%%%%%%%%%%%%%%%%%%%%%%%%%

\subsection{Hamiltonian evolution and total gate complexity}
\label{sec:Ham_evol}
Executing \gls{qpe} requires implementing a quantum circuit to perform the (controlled) Hamiltonian evolution
\mbox{$\op{U} := \exp(-2\pi i \tilde{\op{H}})$},
and, in particular, its powers \mbox{$\op{U}^{2^j} = \exp(-2\pi i 2^j \tilde{\op{H}} )$} for $j = 0, 1, \dots, n-1$.
Given an operation, synthesizing a controlled version proceeds by standard methods \citep[\S4.2]{mikeandike}, and so we will focus on the challenges of implementing the uncontrolled evolution operator.
For most systems, the exact mapping of this propagator to a quantum circuit is not known;
however, there exist different algorithms to approximate such evolution.
Trotterization \footnote{The technique has its origins in the 19th century, but it is known in the physics community for Hale Trotter's work in the 1950s.} is a simple, but well-studied, approach, which approximates the ideal propagator operator of a Hamiltonian given as a sum of simpler Hamiltonians (e.g. Pauli strings), whose generated unitaries are easily expressable on a quantum device.
% $\exp(-2\pi i \mathcal{H} t)$
The basic idea of trotterization is to break down the  (long-)time evolution into series of smaller, more manageable steps, using product formulas. For example, the first-order trotterized evolution of a Hamiltonian $\op{G}=\op{H}_1+\op{H}_2$ is given by
\begin{equation}
e^{{-2\pi i}t\op{G}} = e^{-2\pi i t(\op{H}_1+\op{H}_2)}\approx\left(e^{-2\pi i\frac{t}{r}\op{H}_1}e^{-2\pi i\frac{t}{r}\op{H}_2}\right)^r.
\end{equation}
Here $t$ is the time of the evolution, and $r$ is the number of Trotter steps.

For a generic Hermitian operator $\op{G}=\sum_i\op{H}_i$ and product formula of order $p$, the number of Trotter steps $r$ required to approximate the ideal evolution $\exp\group{-2\pi i t \op{G}}$ to accuracy $\epsilon_T$ is typically bounded by
\begin{equation}
\label{eq:trotter_step}
r = 
\ord\group{
    % \frac{ \norm{\tilde{\op{H}}}^{1+\frac{1}{p}} t^{1+\frac{1}{p}} }{ \epsilon_T^{\frac{1}{p}} }  # "CRUDE" bound
      \frac{ \tilde{\alpha}(\op{G})^{\frac{1}{p}}  t^{1+\frac{1}{p}} }{ \epsilon_T^{\frac{1}{p}} }
},
\end{equation}
where $\tilde{\alpha}(\op{G})$ is derived from a bound on the product formula error. Note that, since the units of $\op{G}$ (energy) and $t$ (time) are inversely proportional, the function $\tilde{\alpha}$ must satisfy $\tilde{\alpha}(c\op{G})=c^{p+1}\tilde{\alpha}(\op{G})$, ensuring that rescaling the units of the Hamiltonian evolution has no effect the final (unitless) result.
One straight-forward such bound is $\tilde{\alpha}(\op{G})\leq||\op{G}||^{p+1}$, resulting in the (crude) approximation
\begin{equation}
\label{eq:trotter_step_crude}
r = 
\ord\group{
    % \frac{ \norm{\tilde{\op{H}}}^{1+\frac{1}{p}} t^{1+\frac{1}{p}} }{ \epsilon_T^{\frac{1}{p}} }  # "CRUDE" bound
      \frac{ \norm{\op{G}}^{\frac{p+1}{p}}  t^{1+\frac{1}{p}} }{ \epsilon_T^{\frac{1}{p}} }
}.
\end{equation}
In \Cref{sec:Improving_Trott} we will discuss a more accurate estimation of the (commutator) scaling $\tilde{\alpha}$.

There are different ways to analyze the error in the ultimate output of \gls{qpe} incurred by the approximate evolution of $\op{U}^{2^j}$, $j=0, \cdots, n-1$, and determining a number of Trotter steps, $r$, sufficient for achieving our desired energy accuracy, $\epsilon_{\nu}$.
We discuss two such ways here. The first looks at how the trotterization errors accumulated in the \gls{qpe} circuit affect its probability of producing the correct desired result, while the second discusses how the eigenvalues of the approximation of $\op{U}$ (which we obtain with \gls{qpe}) relate to the actual energies of interest.

In the first approach, each of the $n$ powers of the evolution operator, $\exp(-2\pi i 2^j \tilde{\op{H}})$,  approximated by these product expansions incur an error of $\epsilon_T$, for a total error in the implemented \gls{qpe} circuit of $n \epsilon_T$ (by, for instance \cite[\S4.5.3]{mikeandike})%
\footnote{This assumes that we can implement the individual terms in the product expansion exactly.
Each term is of the form $\exp\group{-2\pi i 2^j \delta_k \op{P}_k}$, where $P_k$ is a Pauli operator and $\delta_k$ is some real value (determined by the time step of the trotterization and the weight of this particular term).
Again, by standard techniques \citep[\S4.7.3]{mikeandike}, this reduces to being able to implement a single-qubit Z-rotation $\exp\group{-2\pi i 2^j \delta_k Z}$.
% By the Solovay-Kitaev theorem and its extensions, an $\epsilon_k$-accurate approximation of this Z-rotation can be synthesized with overhead polylogarithmic in $\epsilon_k^{-1}$.
}.
The result of using a slightly perturbed circuit to implement \gls{qpe} is that the probability of successfully measuring the correct bitstring (the $n$-bit approximation of the eigenvalue) may be degraded.
However, the change in the success probability is bounded by the error in the circuit implementation;
assuming an ideal success probability of $0.75$ as in the previous subsection, the success probability with the approximate circuit $P_{approx}$ satisfies
$|0.75 - P_{approx}| \le 2 n \epsilon_T$
(by e.g. \cite[Eqn.~4.6.2]{mikeandike}).
Consequently, $\epsilon_T$ should be $\ord(\frac{1}{n})$ (specifically, $\epsilon_T \le \frac{1}{8n}$) to maintain a high ($\ge 0.5$) probability of success.
Recall that $n$ is already given as $n = \log_2\group{8\beta\epsilon_{\nu}^{-1}}$.
We combine this with \Cref{eq:trotter_step};
we take $t = 2^j$ and bound this by $2^{n-1} = 4\beta\epsilon_{\nu}^{-1}$, and we note that $\tilde\alpha(\tilde{\op{H}}) = \tilde{\alpha}\group{\frac{\op{H}}{2\beta}} = (2\beta)^{-(p+1)}\tilde{\alpha}(\op{H})$.
Thus $\tilde\alpha(\tilde{\op{H}})^{\frac{1}{p}} = (2\beta)^{-1-\frac{1}{p}}\tilde{\alpha}(\op{H})^{\frac{1}{p}}$, which gives us
\begin{equation}
\label{eq:trotter_number_1}
% r = \ord\group{ \log_2^{\frac{1}{p}}\group{8\norm{\op{H}}\epsilon_{\nu}^{-1}} }.
% r = \ord\group{ \tilde{\alpha}(\op{H})^{\frac{1}{p}} (2\beta)^{-\frac{p+1}{p}} (4\beta\epsilon_{\nu}^{-1})^{1 + \frac{1}{p}} \log_2^{\frac{1}{p}}\group{8\beta\epsilon_{\nu}^{-1}} }.
r = \ord\group{ \tilde{\alpha}(\op{H})^{\frac{1}{p}} (2\epsilon_{\nu}^{-1})^{1 + \frac{1}{p}} \log_2^{\frac{1}{p}}\group{8\beta\epsilon_{\nu}^{-1}} }.
\end{equation}
This is required for each of the $n$ evolution operators, for a total number of Trotter steps in the \gls{qpe} circuit
\begin{equation}
\label{eq:total_depth_1}
R = \ord\group{ \tilde{\alpha}(\op{H})^{\frac{1}{p}} (2\epsilon_{\nu}^{-1})^{1 + \frac{1}{p}} \group{\log_2\group{8\beta\epsilon_{\nu}^{-1}}}^{1+\frac{1}{p}} }.
\end{equation}

\begin{table*}[t]
    \caption{Complexity Summary for a given Hamiltonian $\op{H}$ }
    \label{tab:complexitysummary}
    \centering
    \begin{tabular}{cccc}
        \hline
        Mapping & Gate complexity & Depth complexity & Number of qubits\\
        \hline
        &&&\\
        Unary
            & $C_G^{\text{unary}} =
                \ord\group{
                    p2 D (Ld^2)^D \tilde{\alpha}(\op{H})^{\frac{1}{p}} (\epsilon_{\nu}^{-1})^{1 + \frac{1}{p}}
                }$
                & $\frac{D}{2L}C_G^{\text{unary}}$
                    & $L\times d$ \\
        &&&\\
        Binary
            & $C_G^{\text{binary}} =
                \ord\group{ C_G^{\text{unary}}\log_2(d)^2 }$
                & $\frac{D}{L}C_G^{\text{binary}}$
                    & $L\times\log_2(d)$ \\
        \hline
    \end{tabular}
\end{table*}

Another approach is as follows.
For some choice of order and number of steps, let us define
$\op{U}_T$
as the product expansion approximating the ideal evolution \mbox{$\op{U} = \exp\group{-2\pi i \tilde{\op{H}}}$}.
$\op{U}_T$ is unitary and can therefore be rewritten as 
$\op{U}_T = \exp(-2\pi i \mathcal{H}')$
where $\mathcal{H}'$ is Hermitian.
Hence, we ``perfectly'' implement the evolution for $\op{H}'$, and performing \gls{qpe} by applying powers of $\op{U}_T$ results in the eigenvalues of $\mathcal{H}'$. 
The error $\epsilon$ in the resulting estimate of the energy is the same order as $\epsilon_T$~\cite{babbush2018, reiher2017}.
This implies that we simply take $\epsilon_T = \ord\group{\frac{\epsilon_{\nu}}{2\beta}}$ in order to achieve the targeted spectroscopic accuracy for the original Hamiltonian $\op{H}$.
Using \Cref{eq:trotter_step} again, where instead we take $t = 1$, the result is that we require
\begin{equation}
\label{eq:trotter_number_2}
r = \ord\group{ (2\beta)^{-1}\tilde{\alpha}(\op{H})^{\frac{1}{p}} (\epsilon_{\nu}^{-1})^{\frac{1}{p}} }.
\end{equation}
In this case, $\op{U}_T$ must be repeated $2^0 + 2^1 + \dots +2^{n-1} = 2^n - 1$ times in the \gls{qpe} circuit for a total number of Trotter steps of
\begin{equation}
\notag
R = \ord\group{ 4 \tilde{\alpha}(\op{H})^{\frac{1}{p}} (\epsilon_{\nu}^{-1})^{1 + \frac{1}{p}} },
\end{equation}
which agrees with \Cref{eq:total_depth_1} up to poly-logarithmic terms.

We now estimate the total gate complexity of \gls{qpe}.
% For the approach leading to estimate~\eqref{eq:trotter_number_1}, each of the $n$ powers of the evolution operator require $r$ Trotter steps
% (we will focus on this approach since it gives a better estimate than the approach leading to \eqref{eq:trotter_number_2}).
The number of terms in each Trotter step depends on the trotterization order $p$ and the number of terms in the Hamiltonian, which in turn depends on the encoding used.
% We will now study the quantum gates requirements per Trotter step. 
In the case of a unary encoding there are $\ord\big((Ld^2)^D\big)$ Pauli terms with locality $2D$.
The number of quantum gates per Trotter step is then $N_g^{\text{unary}} = \ord\big(p2D(Ld^2)^D\big)$.
On the other hand, with a binary encoding, the number of Pauli terms amounts to  $\ord\big((Ld^2)^D\log_2(d)\big)$ and the locality is $\ord(2D\log_2(d))$.
Hence the number of quantum gates in this case is $N_g^{\text{binary}} = \ord\big(p2D(Ld^2)^D\log_2(d)^2\big)$.
In either case, the total gate complexity is $\ord(R N_g^{\cdot})$.
\Cref{tab:complexitysummary} summarizes the complexity for both encodings.
%The query complexity to $\op{U}_T$ is given by $\ord(n2^n)$.
%The total gate complexity, $\ord(n2^nrN_g)$, for both unary and binary encodings are given in \Cref{tab:complexitysummary}.

In practice the depth of the quantum circuit is smaller than the gate complexity since different parts can be executed in parallel. 
For instance, this is the case for trotterized operators acting on different modes.
One can execute in general $L/\mu$ $\mu$-body terms at a time.
Further reductions can be applied in the case of the unary encoding at the modal level since the operators are local.
In this case, the operators acting on the same modes but different modals can also be executed in parallel.
Note that this is not the case in the binary encoding since at the mode level the operators can act on the full set of qubits. 
In the unary encoding this effect accounts for a factor $1/2$ (since the operators always link two modals within one mode).
The final depth complexity is given in \Cref{tab:complexitysummary}.

%%%%%%%%%%%%%%%%%%%%%%%%%%%%%%%%%%%%%%%%%%%%%%%%%%%%
%%%%%%%%%%%%%%%%%%%%%%%%%%%%%%%%%%%%%%%%%%%%%%%%%%%%

\section{Estimation of Trotter errors}
\label{sec:Improving_Trott}

As discussed in \Cref{sec:Ham_evol}, correctly determining
the number of Trotter steps required for a sufficiently accurate simulation of the qubit Hamiltonian is crucial for the efficient implementation of the QPE algorithm.
The crude bound provided in \eqref{eq:trotter_step_crude} generally does not take into account commuting terms in the Hamiltonian and therefore significantly over-estimates the Trotter number.
To obtain more precise results, one can use the commutator-scaling approach described by \citet{childs2021theory}.
In this section we give an overview of the main result from \citet{childs2021theory} and discuss techniques for the efficient evaluation of the relevant quantities.

\subsection{Trotter bounds with commutator scaling}
The qubit Hamiltonian is of the form
\begin{equation}
\mathcal{H}
  = \sum_{i=1}^{N_H} \op{H}_i
  = \sum_{i=1}^{N_H} c_i \op{P}_i \, ,
\end{equation}
where $c_i$ are real coefficients, and 
$\op{P}_i=\sigma_{i,1}\otimes\cdots\otimes\sigma_{i,Ld}$ are Pauli string operators
(and, therefore, $\norm{\op{H}_i}=|c_i|$).
A bound on the Trotter number for a $p$-th order product formula is then given in terms of the commutator scaling
\begin{align}
  \label{eq:comm_scaling}
  \tilde{\alpha} &= \alpha_p(S) 
  \nonumber \\
  & = \sum_{i_0, \cdots, i_p \in S}
  \norm{
    [\op{H}_{i_p},
      [\op{H}_{i_{p-1}},
        [\cdots
          [\op{H}_{i_1},\op{H}_{i_0}]
        ]\cdots
      ]
    ]
  }\, ,   
\end{align}
where $S =\{1, \cdots, {N_H}\}$ is the set of all indices for the terms in the Hamiltonian (as we will see below, it is useful to define commutator scaling as a function acting on a particular subset of indices). In particular, we have
\begin{equation}
\label{eq:scaling_trotter_step}
r = \ord\group{
  \frac{ \tilde{\alpha}^\frac{1}{p}
     t^{1+\frac{1}{p}}
  }{
    \epsilon_T^{\frac{1}{p}}
  }
}.
\end{equation}
If we define $N(S) := \sum_{i \in S} |c_i|$, then the bound \eqref{eq:trotter_step_crude} in \Cref{sec:Ham_evol}
is a direct consequence of \eqref{eq:scaling_trotter_step} and 
the crude bound on the commutator scaling given by
\begin{equation}
  \label{eq:crude_bound}
  \alpha_p(S) \leq N(S)^{p+1} = \left(\sum_{i \in S} |c_i|\right)^{p+1}.
\end{equation}

To obtain more precise bounds, and consequently lower Trotter numbers, one has to examine the commutativity relationships between terms in the Hamiltonian. \Cref{sec:nested_commutator} provides a necessary and sufficient condition for a nested commutator, like the ones in the definition \eqref{eq:comm_scaling}, to be non-zero, as well as details on how to efficiently calculate its exact value.
\Cref{sec:splitting_trick} discusses how we can alleviate the high computational cost of examining \emph{every} nested commutator in \eqref{eq:comm_scaling}.

%%%%%%%%%%%%%%%%%%%%%%%%%%%%%%%%%%%%%%%%%%%%%%%%%%%%%%%%%%%%%%%%%%%
\subsection{Efficient evaluation of nested commutators}
\label{sec:nested_commutator}
We discuss how to calculate a nested commutator of Pauli operators.
As discussed in the previous subsection, nested commutators are important in estimating Trotter numbers.
When the operators involved are (weighted) Pauli string operators, extra structure permits a tidy result on the form of the nested commutator.
Notably, weighted Pauli strings either commute or anti-commute, which is a simple consequence of the (anti-) commutation relations of the basic Pauli operators
(see for instance Equations~(2.74) and (2.75) of \cite{mikeandike}).
Using this fact, the next result implies that a nested commutator of Pauli strings $[\op{P}_p, \dots [\op{P}_1, \op{P}_0]]$ is either zero or $2^p (\op{P}_p \dots \op{P}_1 \op{P}_0)$.

\begin{proposition}
\label{prop:zero_or_not}
Consider a nested commutator of operators that either commute or anti-commute
($\op{P}_i\op{P}_j = \pm \op{P}_j \op{P}_i$ for all $i$, $j$).
If, for all $i \in \set{1,2,\dots, p}$,
$\op{P}_i$ anti-commutes with an odd number of $\op{P}_j$, $j < i$, that is
\begin{align}
  \mathrm{cardinality}\, \group{\set{j : \op{P}_j\op{P}_i = -\op{P}_i\op{P}_j, j<i}} \nonumber \\
  \mod 2 = 1, \, \, \forall i,  \nonumber
\end{align}
then
$[\op{P}_p, \dots [\op{P}_1, \op{P}_0]] = 2^p (\op{P}_p \dots \op{P}_1 \op{P}_0)$.
Otherwise,
$[\op{P}_p, \dots [\op{P}_1, \op{P}_0]] = 0$.
\end{proposition}
\begin{proof}
The proof is by induction;
to start, note that if $\op{P}_1$ and $\op{P}_0$ anti-commute, then 
$[\op{P}_1, \op{P}_0] = \op{P}_1\op{P}_0 - \op{P}_0\op{P}_1 = \op{P}_1\op{P}_0 + \op{P}_1\op{P}_0 = 2\op{P}_1\op{P}_0$.
Otherwise, they commute, and $[\op{P}_1, \op{P}_0] = 0$.
Now, assume that the result holds for $p$.
Assume that for all $i \le p$, $\op{P}_i$ anti-commutes with an odd number of $\op{P}_j$, $j < i$;
then by the induction assumption, it holds that
$[\op{P}_p, \dots [\op{P}_1, \op{P}_0]] = 2^p (\op{P}_p \dots \op{P}_1 \op{P}_0)$.
In this case, note that
\[\begin{aligned}
\relax [\op{P}_{p+1}, [\op{P}_p, \dots [\op{P}_1, \op{P}_0]]]
    = [\op{P}_{p+1}, 2^p \op{P}_p \dots \op{P}_1 \op{P}_0] \\
    = 2^p\big(\op{P}_{p+1} \op{P}_p \dots \op{P}_1 \op{P}_0 - \op{P}_p \dots \op{P}_1 \op{P}_0 \op{P}_{p+1}\big).
\end{aligned}\]
By hypothesis, $\op{P}_{p+1}$ either commutes or anti-commutes with each of $\op{P}_0$, $\op{P}_1, \dots, \op{P}_p$.
Thus, $\op{P}_p \dots \op{P}_1 \op{P}_0 \op{P}_{p+1} = s \op{P}_{p+1} \op{P}_p \dots \op{P}_1 \op{P}_0$,
where $s \in \set{+1, -1}$.
So, if $\op{P}_{p+1}$ anti-commutes with an odd number of them, $\op{P}_p \dots \op{P}_1 \op{P}_0 \op{P}_{p+1} = - \op{P}_{p+1} \op{P}_p \dots \op{P}_1 \op{P}_0$.
Consequently,
\[\begin{aligned}
[\op{P}_{p+1}, [\op{P}_p, \dots [\op{P}_1, \op{P}_0]]] &=  \\
= 2^p\big(\op{P}_{p+1} \op{P}_p \dots &\op{P}_1 \op{P}_0 + \op{P}_{p+1} \op{P}_p \dots \op{P}_1 \op{P}_0\big) \\
&= 2^{p+1}\big(\op{P}_{p+1} \op{P}_p  \dots \op{P}_1 \op{P}_0\big).
\end{aligned}
\]
On the other hand, assume that  for some $i \le p$, $\op{P}_i$ anti-commutes with an even number of $\op{P}_j$, $j < i$;
then by the induction hypothesis, $[\op{P}_p, \dots [\op{P}_1, \op{P}_0]] = 0$.
Therefore $[\op{P}_{p+1}, [\op{P}_p, \dots [\op{P}_1, \op{P}_0]]] = 0$ as well.
If $\op{P}_{p+1}$ anti-commutes with an even number of the previous operators, then $\op{P}_p \dots \op{P}_1 \op{P}_0 \op{P}_{p+1} = \op{P}_{p+1} \op{P}_p \dots \op{P}_1 \op{P}_0$.
Then
$[\op{P}_{p+1}, [\op{P}_p, \dots [\op{P}_1, \op{P}_0]]]
= 2^p\big(\op{P}_{p+1} \op{P}_p \dots \op{P}_1 \op{P}_0 - \op{P}_{p+1} \op{P}_p \dots \op{P}_1 \op{P}_0\big)
= 0$.
The overall result holds for all $p$ by induction.
\end{proof}

If we specifically consider scaled/weighted Pauli strings with nonzero weights, then the conditions of the previous result become necessary and sufficient.

\begin{proposition}
\label{prop:zero_iff}
Consider a nested commutator of weighted Pauli strings with nonzero weights.
Then $[\op{P}_p, \dots [\op{P}_1, \op{P}_0]]=0$
if and only if
for some $i \in \set{1,2,\dots, p}$,
$\op{P}_i$ anti-commutes with an even number of $\op{P}_j$, $j < i$
(that is, $\exists i:$ $\mathrm{cardinality}\group{\set{j : \op{P}_j\op{P}_i = -\op{P}_i\op{P}_j, j<i}} \mod 2 = 0$).
\end{proposition}
\begin{proof}
Assume that for some $i \in \set{1,2,\dots, p}$,
$\op{P}_i$ anti-commutes with an even number of $\op{P}_j$, $j < i$.
Since weighted Pauli strings either commute or anti-commute, we can apply Proposition~\ref{prop:zero_or_not}, and we get $[\op{P}_p, \dots [\op{P}_1, \op{P}_0]] = 0$.
Conversely, assume that the nested commutator is zero.
For a contradiction, assume that for all $i$, $\op{P}_i$ anti-commutes with an odd number of $\op{P}_j$, $j < i$.
Then again by Proposition~\ref{prop:zero_or_not}, $[\op{P}_p, \dots [\op{P}_1, \op{P}_0]] = 2^p (\op{P}_p \dots \op{P}_1 \op{P}_0)$.
As these are nontrivial Pauli strings, $\op{P}_p \dots \op{P}_1 \op{P}_0$ is also a Pauli string with nonzero weight, and thus $\op{P}_p \dots \op{P}_1 \op{P}_0$ is not zero, which is a contradiction.
\end{proof}

Propositions~\ref{prop:zero_or_not} and \ref{prop:zero_iff} imply that each term in an expression like \eqref{eq:comm_scaling} is either $2^p |c_{i_p}c_{i_{p-1}} \cdots c_{i_0}|$ or zero, and confirming whether it is zero requires at most $p(p+1)/{2}$ checks for (anti)-commutativity.

\subsection{Computing the commutator scaling}
\label{sec:splitting_trick}
As the number of terms in \eqref{eq:comm_scaling} is $N_H^{p+1}$, the exact computation of the commutator scaling can be prohibitively expensive even for small molecules.
Nevertheless, bounds on the commutator scaling itself can be obtained by dividing the set $S=\set{1,\dots,N_H}$ into two subsets based on the absolute magnitude of the terms in them, namely,
\begin{align}
 &\set{1, \cdots, N_H} = S_{big} \cup S_{small} = 
 \nonumber \\
& = \set{\; i \;:\; |c_i|>tol\;}
   \cup \set{\;i \; : \; |c_i|\leq tol\;}.   
\end{align}
We then have
\begin{equation}
  \label{eq:first_approx}
  \begin{split}
&\alpha_p({S_{big}}) \leq \alpha_p(S) \leq
\nonumber \\
& \leq \alpha_p(S_{big}) +
{p+1 \choose 1} N(S_{big})^{p}N(S_{small}) +
\cdots \nonumber \\
& \, \, \, \qquad \quad \cdots +{p+1 \choose p+1}N(S_{small})^{p+1} .
  \end{split}
\end{equation}
Note that setting the tolerance $tol$ to zero, amounts to computing the commutator scaling over the full set ($S_{big}=S$) and the upper and lower bounds in \eqref{eq:first_approx} are equal to exact value of $\tilde{\alpha}$. Using $tol=\infty$ on the other hands gives a lower bound of zero ($S_{big}=\emptyset$) and 
the previously mentioned crude upper bound \eqref{eq:crude_bound}.

One can obtain even tighter bounds by applying the same idea
recursively. For example, for bounds on the second order product formula, instead of using
\begin{equation}
  \begin{split}
   &\sum_{\substack{i\in S_{small}\\ {j_1},{j_2} \in S_{big}}}
    \Big\{
    \norm{[\op{P}_i,[\op{P}_{j_1},\op{P}_{j_2}]]} \, + 
    \norm{[\op{P}_{j_1},[\op{P}_i,\op{P}_{j_2}]]} \, + \\ 
    & \hspace{4.2cm} + 
    \norm{[\op{P}_{j_2},[\op{P}_{j_1},\op{P}_i]]}
  \Big\}
  \\ &  \leq  3N(S_{small})N(S_{big})^2,
  \end{split}
\end{equation}
as we did in order to obtain \eqref{eq:first_approx},
one can more carefully examine the commutation of terms
(since the cost of computing $\alpha_1$ is only quadratic in 
the number of terms in a subset), and obtain the stricter bound
\begin{equation}
  \begin{split}
  &\sum_{\substack{i\in S_{small}\\ {j_1},{j_2} \in S_{big}}}
  \Big\{
    \norm{[\op{P}_i,[\op{P}_{j_1},\op{P}_{j_2}]]} \, + 
    \norm{[\op{P}_{j_1},[\op{P}_i,\op{P}_{j_2}]]} \, + \\ 
     & \hspace{2cm} + 
    \norm{[\op{P}_{j_2},[\op{P}_{j_1},\op{P}_i]]}
  \Big\}
  \\ &\leq 
  \Big\{ N(S_{small})\alpha_1(S_{big}) \, +
  \\ & +  N(S_{big})\Big( \alpha_1(S)-\alpha_1(S_{big})-\alpha_1(S_{small}) \Big) \Big\} \leq
  \\ & \leq 
  N(S_{small})\alpha_1(S_{big})
  + \\
  & \hspace{2cm}  + N(S_{big})\Big(\alpha_1(S)-\alpha_1(S_{big}) \Big).
  \end{split}
\end{equation}

Using these techniques, one can obtain reasonably tight bounds on the commutator scaling, and therefore on the Trotter number needed to simulate a particular Hamiltonian to a prescribed accuracy. In \Cref{sec:RE_polyynes} we provide concrete examples of the savings such techniques provide over the conventional estimate \eqref{eq:trotter_step_crude}.

%%%%%%%%%%%%%%%%%%%%%%%%%%%%%%%%%%%%%%%%%%%%%%%%%%%%
%%%%%%%%%%%%%%%%%%%%%%%%%%%%%%%%%%%%%%%%%%%%%%%%%%%%

\section{Quantum resource estimates for polyyne molecules}
\label{sec:RE_polyynes}
\subsection{Motivation}
\label{sec:polyynes_motivation}

Many of the commercial refining and chemical processes are designed to operate at low temperatures and often in presence of catalysts. This aims to minimize the material challenges and cost of the reactor, and reduce the carbon footprint of the process. The operating temperature of a process is indicative of the extent of possible vibrational excitations in a polyatomic molecule.
%
%Many of the commercial refining and chemical processes are designed to operate at low temperatures and often in presence of catalysts. This aims to minimize the material challenges and cost of the reactor, and also in an effort to reduce the carbon footprint of the process. The operating temperature of a process is indicative of the extent of possible vibrational excitations in a polyatomic molecule.
%The probability $p_n$ of a polyatomic molecule to be in a given vibrational state $n$ with energy $E_n$ follows a Boltzmann distribution, $p_n \propto e^{-E_n/kT}$ where $k$ is Boltzmann's constant and $T$ is the operating temperature. The latter is therefore indicative of the extent of possible vibrational excitations in the molecule. 
In delayed coking, for example, the commercial thermal process for the decomposition of heavy hydrocarbons to coke is operated at temperatures below 723 K. This suggests an accurate description of the first four or five vibrational energy levels would be needed to better describe the contribution of vibrational excitation to reactivity.
\Gls{fccvd} is an emerging technology operating at similar temperatures. 
\gls{fccvd} thermocatalytically upgrades light hydrocarbon gases (e.g. ethane, propane, ethylene or acetylene) to a more value-added \gls{cnt} or \gls{cf} in catalytic chemical reactions with iron nanoparticles as catalysts. 
\glspl{cnt} possess various interesting chemical and physical properties such as high tensile strength, light weight, and high electrical and thermal conductivity~\cite{dresselhaus2000carbon}.
\glspl{cf} are used as filler materials in composites to improve their mechanical and thermal properties, and carbon composites are an attractive way of sequestering the pyrolysis carbon as infrastructure material~\cite{kinloch2018composites}.
%An interesting emerging technology operating at similar temperatures is the \gls{fccvd}. \gls{fccvd} upgrades light hydrocarbon gasses (e.g. ethane, propane, ethylene or acetylene) to a more value-added \glspl{cnt}.
%\glspl{cnt} possess various interesting chemical and physical properties such as high tensile strength, ultra-light as well as high electrical and thermal conductivity.~\cite{dresselhaus2000carbon}
%Carbon fibers can for instance be used as filler materials in composites to improve their mechanical and thermal properties.~\cite{kinloch2018composites}

Both gas-phase high-temperature thermal chemistry, as well as catalytic chemistry on nanoparticles, involve rod-like polyyne molecules as intermediates en-route to nano-graphenic templates and \glspl{cnt}.
% As a first step in the \gls{fccvd} process, the light hydrocarbon gases are converted to polyynes which are in turn converted to graphene and then to \gls{cnt}.
Polyynes are organic molecules containing alternating carbon single and triple bonds, $(\ce{-C#C-})_n$ with $n>1$.
They belong to very high $D_{\infty h}$ symmetry point group with a large number of low frequency degenerate wagging and degenerate bending modes. 
For example, the degenerate CCC bend frequency is 231~cm$^{-1}$ in 1,3-butadiyne~\cite{NISTWebBook}
%HCCCCH
and can decrease further with increasing the carbon number in polyynes. 
%They belong to the very high $D_{\infty h}$ symmetry point group with a large number of degenerate low frequency wagging and bending modes. 
An accurate description of such low frequency vibrations is essential to understand their reactivity, in particular to obtain the vibrational partition functions necessary to calculate reaction rate constants using transition state theory.
This is challenging in practice, as such high symmetry point groups are characterized by many silent modes that cannot be detected via infrared or Raman spectroscopy, nor are accessible via experiment since polyynes are challenging to stabilize~\cite{chalifoux2010synthesis, milani2011charge}.
Nevertheless, the low-frequency modes are important for accurately calculating the thermodynamic properties.
This reinforces the need for computational techniques which can accurately predict energies of vibrational states and marks vibrational structure calculations as a potential industrial application for quantum computing.

{}

\subsection{Quantum memory footprint}
As the polyyne molecules we consider have linear geometry, the number of modes, $L$, is easily seen to be 
\begin{equation}
  \begin{split}
  \label{eq:num_modes}
   L &= 3N_{atoms}-5 = 3N_{carbon\ atoms}+1 
  \\
  & = 6N_{triple\ bonds}+1.
  \end{split}
\end{equation}
Thus, acetylene, diacetylene, and tri-acetylene have 7, 13, and 19 vibrational modes, respectively.
The number of qubits required for unary and binary mappings, across varying chains of triple bonds in polyyne molecules, depends on the chosen number of modals for each vibrational mode.
Assuming a constant number of modals per mode, the smallest system considered -- a polyyne with one triple bond and 4 modals per mode -- requires 28 qubits for the unary mapping and 14 for the binary mapping.
Similarly, for a larger system, such as a polyyne with 80 triple bonds and 10 modals per mode, 4810 qubits are needed for the unary mapping and 1598 for the binary mapping.

% \begin{figure}
% \centering
% \begin{subfigure}{.5\textwidth}
%   \centering
%   \includegraphics[width=1\linewidth]{qubit_count_unary.pdf}
%   %\caption{A subfigure}
%   \label{fig:qubit_count_unary}
% \end{subfigure}%
% \begin{subfigure}{.5\textwidth}
%   \centering
%   \includegraphics[width=1\linewidth]{qubit_count_binary.pdf}
%   %\caption{B subfigure}
%   \label{fig:qubit_count_binary}
% \end{subfigure}
% \caption{Qubit count in the logarithmic scale using the unary mapping (left) and binary mapping (right) for Polyynes with $k$ number of triple bonds labelled as Pol$k$ for varying number of modals per vibrational mode.}
% \label{fig:qubit count}
% \end{figure}

%\begin{figure}
%    \centering
%    \includegraphics[width = 0.49\textwidth]{qubit_count_combined.pdf}
%    \caption{Qubit count in the logarithmic scale using the unary mapping (left) and binary mapping (right) for polyynes with $k$ number of triple bonds labelled as Pol$k$ for varying number of modals per vibrational mode.}
%    \label{fig:qubit count}
 % \end{figure}

\subsection{Hamiltonian size}

For the remainder of this section, we will set the truncation order of our vibrational Hamiltonian~\eqref{eq:Hamil2ndquan} to $D=3$, which is considered an optimal choice, striking a balance between the desired accuracy and the computational cost required for generating the Hamiltonian~\cite{PhysRevA.104.062419, CARTER19971179, Rauhut2004}.
The asymptotic bound \eqref{eq:num_terms_asymp} on the number of terms in the second-quantized Hamiltonian counts all summands in \eqref{eq:Hamil2ndquan}, 
including those that evaluate to zero. 
While only a fraction of the total number, the number of non-zero coefficients in \eqref{eq:Hamil2ndquan} can quickly explode even for small molecules with reasonable number of modes ($L=7,13,19,\cdots$ in our case) and a limited number of modal functions ($d=4,6,8,10$).\footnote{Possible errors from the choice of basis sets, modal truncation, and related factors have been extensively discussed in the literature. For example, works such as \cite{Sumathy2007, Zapata2021, HeadGordon2023} provide comprehensive analyses of these issues.}
Table~\ref{tab:num_terms} lists the number of non-zero terms $N_{H,\neq0}$ in the second-quantized Hamiltonian in the VSCF basis, as well the number of terms greater than $10^{-8}$~cm$^{-1}$, which will be our cutoff threshold in the qubit Hamiltonian.
%%%%%%%%%%%%%%%%%%%%%%%%%%%%%%%%%%%%%%%%%%%%%%%%%%%%%%%%%%%%
\begin{table}
  \caption{Number of terms in the second-quantized Hamiltonian for the polyyne molecules of consideration. $N_H$ is the number of all summands in \eqref{eq:Hamil2ndquan}; $N_{H,\neq0}$ denotes the number of non-zero terms; and $N_{H,\geq10^{-8}}$ denotes the number of terms that exceed $10^{-8}$ by absolute value. Note that the numbers here are for the unary encoding.}
  \label{tab:num_terms}
  \centering
  \resizebox{\columnwidth}{!}{%
  \begin{tabular}{lcrrr}
    \hline
    Molecule & Modals ($d$) & $N_H$ & $N_{H,\neq0}$ & $N_{H,\geq10^{-8}}$ \\
    \hline
    acetylene & 4 & 148848 & 37567 & 15656 \\
    acetylene & 6 & 1660428 & 540684 & 163386 \\
    diacetylene & 4 & 1191632 & 327744 & 117504 \\
    diacetylene & 6 & 13445172 & 4860468 & 1277408\\
    triacetylene & 4 & 4013104 & 1255312 & 392952\\
    triacetylene & 6 & 45431964 & 17989164 & 4290852\\
    tetraacetylene & 4 & 9498000 & 3552748 & 916982 \\
    \hline
  \end{tabular}%
  }
  \end{table}

%%%%%%%%%%%%%%%%%%%%%%%%%%%%%%%%%%%%%%%%%%%%%%%%%%%%%%%%%%%%%

\subsection{Trotterization cost}
To keep the number of terms in the qubit Hamiltonian manageable, we favor the unary encoding
\eqref{eq:ONV_Christiansen} versus the binary one\footnote{While this increases the number of qubits we need, the effect is not substantial, since a small number of modal basis functions $d \le 10$ is sufficient in the context discussed in \Cref{sec:polyynes_motivation}} and further discard from the qubit Hamiltonian any terms with absolute value less than $10^{-8}$.
This results in qubit Hamiltonian operators with number of Pauli strings roughly equal to the values of $N_{H,\ge 10^{-8}}$ reported in \Cref{tab:num_terms}.
While these numbers are significantly less than the total number of terms $N_H$, they still grow as expected, namely as (a fraction of) $L^3d^6$ (since we have fixed $D=3$).
Thus, the techniques of \Cref{sec:splitting_trick} are necessary to accurately estimate the trotterization cost for these operators.
\Cref{fig:second_order_trotter} reports the total number of Trotter steps needed to guarantee an accuracy $\epsilon_{\nu}=1~cm^{-1}$, in the energies obtained with the workflow described in \Cref{sec:QC_algo}, for the second order ($p=2$) trotterization formula \footnote{For the case $p=2$, the small constant pre-factors can be tracked akin to the analysis in Section 5.1 of~\cite{childs2021theory}. It should be noted that the numbers we report here are still only upper bounds - due to the applications of the triangle inequality - for the number of Trotter steps required.}.
The values in red correspond to bounds obtained with \eqref{eq:crude_bound}, while values in green are obtained with the techniques described in \Cref{sec:splitting_trick}.

As expected there are significant savings in estimating the Trotter number more accurately. What is more important is that the growth of the green curve in \Cref{fig:second_order_trotter} seems to be of lower-order, suggesting that the savings are more than just a constant factor.
Indeed, fitting a polynomial of the type $AL^3+BL^2$, which is the expected growth in \eqref{eq:num_terms_asymp}, results in a negative leading coefficient for the values in green, obtained using the commutator-scaling approach.This suggests that these bounds grow only quadratically and the savings when compared to the crude bounds in red (which do grow cubically) will be even more noticeable for larger molecules.

\subsection{Circuit parallelization}
% Just delete this first paragraph?
% Or fill in the concrete numbers implied by Table 2?
The concrete size of the Hamiltonian and the total number of Trotter steps in the \gls{qpe} implementation, discussed in the subsections above, allow us to directly compute the number of gates $N_{g}$ for the Hamiltonian simulation part of the \gls{qpe} circuit
\begin{equation}
  N_{g} = \frac{gates} {per\ term}\times\frac{terms} {per\ Trotter\ step}\times \,  N_{Trott}\, , 
\end{equation}
where $N_{Trott}$ is the total number of Trotter steps required. 
An accurate estimate for the cost of Hamiltonian evolution in \gls{qpe}, along with the number of extra ancilla qubits chosen for the accuracy of the algorithm, allows one to fully describe the complexity for the full \gls{qpe} circuit in \Cref{fig:qpe}.

\begin{figure}
  \centering
  \includegraphics[width=0.5\textwidth]{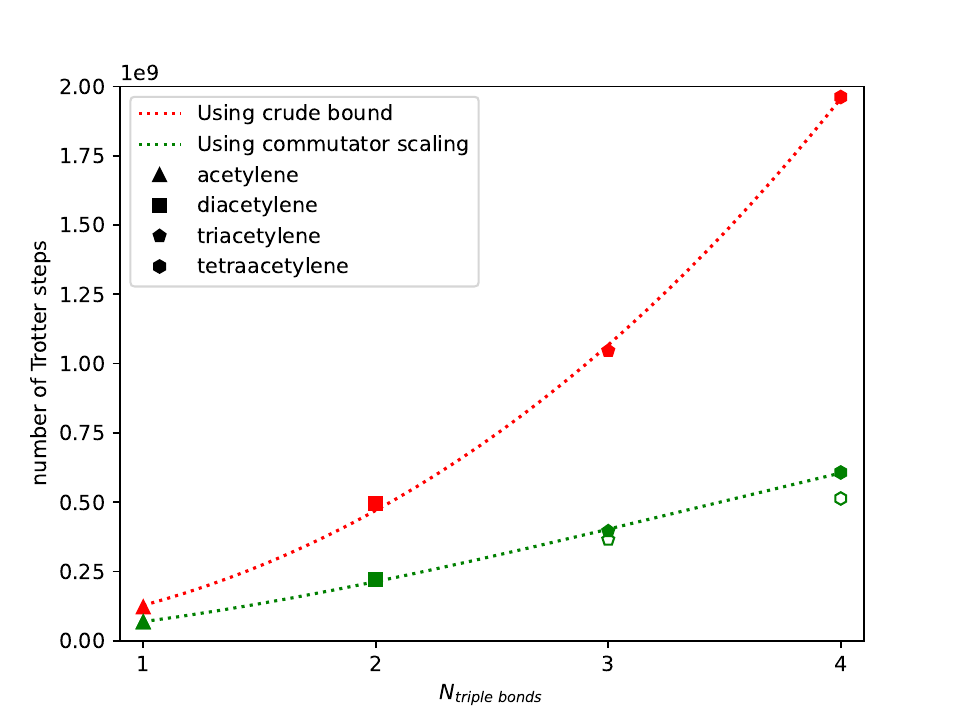}
  \caption{
  %\todo{Axis labels.} 
  Estimates for the total number of Trotter steps required to achieve accuracy of $\epsilon_{\nu}=1~cm^{-1}$ for the first few polyyne molecules of interest. Crude bound estimates are plotted in red and bounds using commutator scaling - in green. The solid green symbols represent an upper bound, while lower bounds are plotted with green symbol outlines. The dashed lines plot a fit to a polynomial of the type $AL^3+BL^2 = A(6N_{triple\ bonds}+1)^3+B(6N_{triple\ bonds}+1)^2$, as expected from our asymptotic results.
  It should be noted that the best fit for the values obtained with the commutator scaling has a negative leading coefficient, suggesting that, when commutations are accounted for, the Trotter number grows only quadratically.}
  \label{fig:second_order_trotter}
\end{figure}

The number of gates is not the sole metric for complexity. Another factor to consider in estimating the resources required for simulating each Trotter step is the depth of the circuit. Using a greedy algorithm as discussed below, we can estimate a ratio of the total number of Paulis to be simulated in each Trotter step to the number of Paulis that need to be simulated sequentially due to having overlapping support on the qubits.  This ratio is a reflection of the effective reduction in circuit depth that can be obtained due to the different localities observed in the Pauli strings of the Hamiltonian. 

To estimate such ratio, a Pauli string from the pool of Pauli strings to be used in the Trotter expansion is drawn at random and the set of qubits over which it acts non-trivially is compared to that of another randomly drawn Pauli string. If two sets of qubits are disjoint, then these Pauli strings can be executed simultaneously in a quantum circuit leading to no increase in depth. If the qubit sets are overlapping then the count of layers to be executed sequentially is increased by one. This procedure is continued until the Pauli pool is exhausted. A specific value for the speedup is then obtained by taking the ratio of total number of Paulis to the number of Pauli layers that need to be executed sequentially. The average expected speedup from such greedy algorithm can be computed by repeating the procedure multiple times and averaging the results.

Below we will compare these numbers for the electronic and vibrational structure problems on the acetylene molecule. Since the maximum locality of the Pauli strings in vibrational structure Hamiltonian is governed by the degree of truncation, it is independent of the number of modals in the problem. This leads to a favourable estimate of roughly 1.5 Pauli strings per unit of circuit depth for the vibrational structure Hamiltonian of the acetylene molecule truncated at the 3rd order with different number of modals. This ratio increases as the number of modals increases as expected. In comparison, the maximum locality of the Pauli strings in the electronic structure problem for most commonly used fermion-to-qubit mappings is equal to the number of qubits. On average this leads to fewer Pauli strings that can be implemented simultaneously in a quantum circuit. For the acetylene electronic structure Hamiltonian in different basis sets we find that roughly 1.1 Pauli strings can be executed at the expense of the resources required for simulation of 1 Pauli string. The exact numbers obtained from the implementation of the greedy algorithm discussed above are presented in Table \ref{tab: Depth_speedup}.

\begin{table}[h]
\caption{Ratio of total number of Paulis to number of Paulis that need to be simulated sequentially averaged over 100 runs, each with a different random order of selecting Paulis for the electronic and vibrational Hamiltonian of acetylene molecule.}
\label{tab: Depth_speedup}
\begin{center}
\resizebox{\columnwidth}{!}{%
\begin{tabular}{cccc} 
 \hline
 Electronic problem & Ratio & Vibrational problem & Ratio \\ [0.5ex] 
 \hline
 sto-3g & 1.14539 & 4 modals & 1.34023 \\ 
 
 6-31g & 1.15238 & 6 modals & 1.46114 \\
 cc-pvdz & 1.16602 & 8 modals & 1.60766 \\
 \hline
\end{tabular}%
}
\end{center}
\end{table}

%%%%%%%%%%%%%%%%%%%%%%%%%%%%%%%%%%%%%%%%%%%%%%%%%%%%
%%%%%%%%%%%%%%%%%%%%%%%%%%%%%%%%%%%%%%%%%%%%%%%%%%%%

\section{Summary and conclusion}
\label{sec:Conclusion}
We have provided an accurate asymptotic analysis of quantum resources required for simulating molecular vibrational structure problems on quantum computers. 
We have focused on an explicit Hamiltonian model, the L-mode representation of the Hamiltonian.
The estimates for the quantum resources required for this Hamiltonian representation have been given in Table~\ref{tab:complexitysummary}.
We have used the commutator-scaling approach~\cite{childs2021theory} and devised an efficient scheme to estimate the nested commutators. This has led to a tight(er) upper bound on the Trotter number requisite for simulating a particular Hamiltonian, for a given accuracy. 
Employing this technique, we have provided a detailed quantitative study of the quantum computational cost for concrete examples, namely, polyyne molecules. 

%#DT
Overall results suggest that the combined quantum resources required for achieving quantum advantage for vibrational structure problems might be lower than those for the electronic structure problems, which is in agreement with~\cite{PhysRevA.104.062419}. Since Pauli strings for the vibrational qubit Hamiltonian are more localized, they are expected to be easier and faster to simulate than those in the electronic qubit Hamiltonian. In addition, our estimate of the number of Pauli strings that can be simultaneously executed in a quantum circuit when implementing a single Trotter step is larger for the vibrational structure problem as compared to the electronic one. This may give the molecular vibrational simulations an edge over the electronic simulations due to higher locality of Pauli strings.

Nevertheless, further tailored studies are needed to truly compare the quantum computational cost of vibrational and electronic structure problems. Just as \textit{basis functions} are employed in electronic structure, \textit{modals} are utilized in the context of the vibrational structure simulations for the same purpose, see \eqref{eq:Hamil2ndquan}. Often, fewer modals are needed than basis functions. Nevertheless, as shown in Table~\ref{tab:complexitysummary}, their number greatly affects the required quantum resources. Furthermore, the number of modals required for a vibrational structure problem increases with temperature.
While the number of relevant terms in the vibrational qubit Hamiltonian may be smaller from an asymptotic perspective, the situation is not as clean-cut for molecules that we can reasonably expect to be able to study on quantum devices in the near term, see Table~5 in the \textit{Supplementary Information}. 
Determining the chemical components and properties where either vibrational or electronic structure simulation can lead to quantum advantage requires a thorough case-by-case study.

It should be noted that alternative approaches to Hamiltonian simulation exist beyond product formulas/Trotter-Suzuki formula. Notably, methods based on the Linear Combination of Unitaries (LCU)~\cite{LCU2012, LCU2024} and qubitization~\cite{Qubitization2019} offer promising alternatives. Our focus, however, is on the Trotter-Suzuki formula due to its extensive study in the literature for electronic structure problems and the fact that it provides loose bounds for resource estimates compared to the actual requirements. Additionally, it does not necessitate extra ancilla qubits. While approaches like LCU and qubitization have demonstrated lower or even optimal theoretical asymptotic computational costs, they demand significantly more computational resources in terms of ancilla qubits and control operations than product formulas, particularly for electronic and vibrational Hamiltonians with large number of terms in the Hamiltonian. For instance, the number of ancillas required for LCU-based approaches scales as $log(N_H)$, potentially leading to a substantial increase in physical qubits and gate operations for fault-tolerant architectures, see~\cite{qcentricmaterial} for more details.
Moreover, a systematic comparison of the pre-factors involved in these methods for specific problems is necessary, which is beyond the scope of the present work.

In light of the discussion presented in this paper, in our opinion, it is not yet clear what is the best frame for comparing molecular electronic and vibrational structure problems, in order to determine which one would offer quantum advantage first. Rather than comparing the quantum resource requirements for problems resulting in similarly sized Hilbert spaces, one can, for example, compare it for selected sets of small challenging problems from each domain (and their respective chemical accuracy).
Considering that the electronic structure problems on quantum computers are somewhat well-understood, however, for such comparison to be meaningful we will need a better understanding of possible optimizations that could improve the qubit numbers and the circuit depth for the vibrational structure simulations as well.
%#DT 

%Overall results suggest that the combined quantum resources required for achieving quantum advantage for vibrational structure problems might be lower than for those the electronic structure problems. More specifically, we expect the circuits implementing QPE for the vibrational level calculations to be of lower depth, when compared to those for the electronic energy calculations, for relevant molecules mapped to the same (or comparable) number of qubits.

%%%%%%%%%%%%%%%%%%%%%%%%%%%%%%%%%%%%%%%%%%%%%%%%%%%%
%%%%%%%%%%%%%%%%%%%%%%%%%%%%%%%%%%%%%%%%%%%%%%%%%%%%

\section*{Acknowledgments}
We gratefully acknowledge Mario Motta, Ivano Tavernelli, Gavin Jones, Ieva Liepuoniute, Alberto Baiardi, Spencer T Stober, and Panagiotis Kl Barkoutsos for helpful discussions and critical comments on the manuscript.

%%%%%%%%%%%%%%%%%%%%%%%%%%%%%%%%%%%%%%%%%%%%%%%%%%%%
%%%%%%%%%%%%%%%%%%%%%%%%%%%%%%%%%%%%%%%%%%%%%%%%%%%%
%%%%%%%%%%%%%%%%%%%%%%%%%%%%%%%%%%%%%%%%%%%%%%%%%%%%
\bibliographystyle{unsrtnat}
\bibliography{references}

\begin{thebibliography}{54}
\providecommand{\natexlab}[1]{#1}
\providecommand{\url}[1]{\texttt{#1}}
\expandafter\ifx\csname urlstyle\endcsname\relax
  \providecommand{\doi}[1]{doi: #1}\else
  \providecommand{\doi}{doi: \begingroup \urlstyle{rm}\Url}\fi

\bibitem[Marzari et~al.(2021)Marzari, Ferretti, and Wolverton]{Marzari2021}
Nicola Marzari, Andrea Ferretti, and Chris Wolverton.
\newblock Electronic-structure methods for materials design.
\newblock \emph{Nature Materials}, 20\penalty0 (6):\penalty0 736--749, 2021.
\newblock \doi{https://doi.org/10.1038/s41563-021-01013-3}.

\bibitem[Helgaker et~al.(2000)Helgaker, J\o{}rgensen, and Olsen]{MESTheory}
Trygve Helgaker, Poul J\o{}rgensen, and Jeppe Olsen.
\newblock \emph{Molecular Electronic‐Structure Theory}.
\newblock John Wiley \& Sons, Ltd, 2000.
\newblock ISBN 9781119019572.
\newblock \doi{https://doi.org/10.1002/9781119019572}.

\bibitem[Leach(2001)]{Leach_CC}
Andrew~R. Leach.
\newblock \emph{Molecular Modelling: Principles and Applications}.
\newblock Pearson, 2001.
\newblock ISBN 978-0582382107.
\newblock \doi{https://doi.org/10.1021/ci9804241}.

\bibitem[Cramer(2004)]{Cramer_CC}
Cristopher~J. Cramer.
\newblock \emph{Essentials of Computational Chemistry: Theories and Models}.
\newblock John Wiley \& Sons, Ltd, 2004.
\newblock ISBN 978-0-470-09182-1.
\newblock \doi{https://doi.org/10.1021/ci010445m}.

\bibitem[Jensen(2017)]{Jensen_CC}
Frank Jensen.
\newblock \emph{Introduction to Computational Chemistry}.
\newblock John Wiley \& Sons, Ltd, 2017.
\newblock ISBN 978-1-118-82599-0.

\bibitem[Aspuru-Guzik et~al.(2005)Aspuru-Guzik, Dutoi, Love, and
  Head-Gordon]{doi:10.1126/science.1113479}
Alán Aspuru-Guzik, Anthony~D. Dutoi, Peter~J. Love, and Martin Head-Gordon.
\newblock Simulated quantum computation of molecular energies.
\newblock \emph{Science}, 309\penalty0 (5741):\penalty0 1704--1707, 2005.
\newblock \doi{https://doi.org/10.1126/science.1113479}.

\bibitem[Liu et~al.(2022)Liu, Low, Steiger, H{\"a}ner, Reiher, and
  Troyer]{Liu2022}
Hongbin Liu, Guang~Hao Low, Damian~S. Steiger, Thomas H{\"a}ner, Markus Reiher,
  and Matthias Troyer.
\newblock Prospects of quantum computing for molecular sciences.
\newblock \emph{Materials Theory}, 6\penalty0 (1):\penalty0 11, 2022.
\newblock \doi{https://doi.org/10.1186/s41313-021-00039-z}.

\bibitem[Cao et~al.(2019)Cao, Romero, Olson, Degroote, Johnson, Kieferov{\'a},
  Kivlichan, Menke, Peropadre, Sawaya, Sim, Veis, and Aspuru-Guzik]{Cao2019}
Yudong Cao, Jonathan Romero, Jonathan~P. Olson, Matthias Degroote, Peter~D.
  Johnson, M{\'a}ria Kieferov{\'a}, Ian~D. Kivlichan, Tim Menke, Borja
  Peropadre, Nicolas P.~D. Sawaya, Sukin Sim, Libor Veis, and Al{\'a}n
  Aspuru-Guzik.
\newblock Quantum chemistry in the age of quantum computing.
\newblock \emph{Chemical Reviews}, 119\penalty0 (19):\penalty0 10856--10915,
  2019.
\newblock \doi{https://doi.org/10.1021/acs.chemrev.8b00803}.

\bibitem[Lanyon et~al.(2010)Lanyon, Whitfield, Gillett, Goggin, Almeida,
  Kassal, Biamonte, Mohseni, Powell, Barbieri, Aspuru-Guzik, and
  White]{Lanyon2010}
Benjamin~P. Lanyon, James~D. Whitfield, Geoff~G. Gillett, Michael~E. Goggin,
  Marcelo~P. Almeida, Ivan Kassal, Jacob~D. Biamonte, Masoud Mohseni, Ben~J.
  Powell, Marco Barbieri, Alán Aspuru-Guzik, and Andrew~G. White.
\newblock Towards quantum chemistry on a quantum computer.
\newblock \emph{Nature Chemistry}, 2\penalty0 (2):\penalty0 106--111, 2010.
\newblock \doi{https://doi.org/10.1038/nchem.483}.

\bibitem[Kandala et~al.(2017)Kandala, Mezzacapo, Temme, Takita, Brink, Chow,
  and Gambetta]{Kandala2017}
Abhinav Kandala, Antonio Mezzacapo, Kristan Temme, Maika Takita, Markus Brink,
  Jerry~M. Chow, and Jay~M. Gambetta.
\newblock Hardware-efficient variational quantum eigensolver for small
  molecules and quantum magnets.
\newblock \emph{Nature}, 549\penalty0 (7671):\penalty0 242--246, 2017.
\newblock \doi{https://doi.org/10.1038/nature23879}.

\bibitem[Peruzzo et~al.(2014)Peruzzo, McClean, Shadbolt, Yung, Zhou, Love,
  Aspuru-Guzik, and O'Brien]{Peruzzo2014}
Alberto Peruzzo, Jarrod McClean, Peter Shadbolt, Man-Hong Yung, Xiao-Qi Zhou,
  Peter~J. Love, Al{\'a}n Aspuru-Guzik, and Jeremy~L. O'Brien.
\newblock A variational eigenvalue solver on a photonic quantum processor.
\newblock \emph{Nature Communications}, 5\penalty0 (1):\penalty0 4213, 2014.
\newblock \doi{https://doi.org/10.1038/ncomms5213}.

\bibitem[Cerezo et~al.(2021)Cerezo, Arrasmith, Babbush, Benjamin, Endo, Fujii,
  McClean, Mitarai, Yuan, Cincio, and Coles]{Cerezo2021}
M.~Cerezo, Andrew Arrasmith, Ryan Babbush, Simon~C. Benjamin, Suguru Endo,
  Keisuke Fujii, Jarrod~R. McClean, Kosuke Mitarai, Xiao Yuan, Lukasz Cincio,
  and Patrick~J. Coles.
\newblock Variational quantum algorithms.
\newblock \emph{Nature Reviews Physics}, 3\penalty0 (9):\penalty0 625--644,
  2021.
\newblock \doi{https://doi.org/10.1038/s42254-021-00348-9}.

\bibitem[Reiher et~al.(2017)Reiher, Wiebe, Svore, Wecker, and
  Troyer]{reiher2017}
Markus Reiher, Nathan Wiebe, Krysta~M. Svore, Dave Wecker, and Matthias Troyer.
\newblock Elucidating reaction mechanisms on quantum computers.
\newblock \emph{Proceedings of the national academy of sciences}, 114\penalty0
  (29):\penalty0 7555--7560, 2017.
\newblock \doi{https://doi.org/10.1073/pnas.1619152114}.

\bibitem[Laidler(1987)]{Laidler}
Keith~James Laidler.
\newblock \emph{Chemical Kinetics}.
\newblock Prentice Hall, 3rd Revised ed edition, 1987.
\newblock ISBN 978-0060438623.

\bibitem[Sumathi and Green~Jr.(2002)]{Sumathi2002}
R.~Sumathi and William~H. Green~Jr.
\newblock A priori rate constants for kinetic modeling.
\newblock \emph{Theoretical Chemistry Accounts}, 108:\penalty0 187--213, 2002.
\newblock \doi{https://doi.org/10.1007/s00214-002-0368-4}.

\bibitem[Sawaya et~al.(2021)Sawaya, Paesani, and Tabor]{PhysRevA.104.062419}
Nicolas P.~D. Sawaya, Francesco Paesani, and Daniel~P. Tabor.
\newblock Near- and long-term quantum algorithmic approaches for vibrational
  spectroscopy.
\newblock \emph{Phys. Rev. A}, 104:\penalty0 062419, 2021.
\newblock \doi{https://doi.org/10.1103/PhysRevA.104.062419}.

\bibitem[Sawaya et~al.(2020)Sawaya, Menke, Kyaw, Johri, Aspuru-Guzik, and
  Guerreschi]{sawayanpj2020}
Nicolas P.~D. Sawaya, Tim Menke, Thi~Ha Kyaw, Sonika Johri, Alán Aspuru-Guzik,
  and Gian~Giacomo Guerreschi.
\newblock Resource-efficient digital quantum simulation of d-level systems for
  photonic, vibrational, and spin-s hamiltonians.
\newblock \emph{npj Quantum Information}, 6\penalty0 (49), 2020.
\newblock \doi{https://doi.org/10.1038/s41534-020-0278-0}.

\bibitem[Sawaya and Huh(2019)]{sawayaJPCL2019}
Nicolas P.~D. Sawaya and Joonsuk Huh.
\newblock Quantum algorithm for calculating molecular vibronic spectra.
\newblock \emph{J. Phys. Chem. Lett.}, 10\penalty0 (13):\penalty0 3586–3591,
  2019.
\newblock \doi{https://doi.org/10.1021/acs.jpclett.9b01117}.

\bibitem[Ollitrault et~al.(2020)Ollitrault, Baiardi, Reiher, and
  Tavernelli]{ollitraultCS2020}
Pauline~J. Ollitrault, Alberto Baiardi, Markus Reiher, and Ivano Tavernelli.
\newblock Hardware efficient quantum algorithms for vibrational structure
  calculations.
\newblock \emph{Chem. Sci.}, 11:\penalty0 6842--6855, 2020.
\newblock \doi{https://doi.org/10.1039/D0SC01908A}.

\bibitem[Sparrow et~al.(2018)Sparrow, Martín-López, Maraviglia, Neville,
  Harrold, Carolan, Joglekar, Hashimoto, Matsuda, O’Brien, Tew, and
  Laing]{sparrowNature2018}
C.~Sparrow, E.~Martín-López, N.~Maraviglia, A.~Neville, C.~Harrold,
  J.~Carolan, Y.~N. Joglekar, T.~Hashimoto, N.~Matsuda, J.~L. O’Brien, D.~P.
  Tew, and A.~Laing.
\newblock Simulating the vibrational quantum dynamics of molecules using
  photonics.
\newblock \emph{Nature}, 557\penalty0 (7707):\penalty0 660--667, 2018.
\newblock \doi{https://doi.org/10.1038/s41586-018-0152-9}.

\bibitem[Magann et~al.(2021)Magann, Grace, Rabitz, and Sarovar]{magannPRR2021}
Alicia~B. Magann, Matthew~D. Grace, Herschel~A. Rabitz, and Mohan Sarovar.
\newblock Digital quantum simulation of molecular dynamics and control.
\newblock \emph{Phys. Rev. Research}, 3:\penalty0 023165, 2021.
\newblock \doi{https://doi.org/10.1103/PhysRevResearch.3.023165}.

\bibitem[Majland et~al.(2023)Majland, Jensen, Højlund, Zinner, and
  Christiansen]{majland2022}
Marco Majland, Rasmus~B. Jensen, Mads~G. Højlund, Nikolaj~T. Zinner, and Ove
  Christiansen.
\newblock Optimizing the number of measurements for vibrational structure on
  quantum computers: coordinates and measurement schemes.
\newblock \emph{Chem. Sci.}, 14:\penalty0 7733--7742, 2023.
\newblock \doi{https://doi.org/10.1039/D3SC01984E}.

\bibitem[Jahangiri et~al.(2020)Jahangiri, Arrazola, Quesada, and
  Delgado]{jahangiriPCCP2020}
Soran Jahangiri, Juan~M. Arrazola, Nicolás Quesada, and Alain Delgado.
\newblock Quantum algorithm for simulating molecular vibrational excitations.
\newblock \emph{Phys. Chem. Chem. Phys.}, 22:\penalty0 25528--25537, 2020.
\newblock \doi{https://doi.org/10.1039/D0CP03593A}.

\bibitem[Richerme et~al.(2023)Richerme, Revelle, Saha, Lopez-Ruiz, Dwivedi,
  Norrell, Yale, Lobser, Burch, Clark, Smith, Sabry, and Iyengar]{richerme2023}
Philip Richerme, Melissa~C. Revelle, Debadrita Saha, Miguel~A. Lopez-Ruiz,
  Anurag Dwivedi, Sam~A. Norrell, Christopher~G. Yale, Daniel Lobser, Ashlyn~D.
  Burch, Susan~M. Clark, Jeremy~M. Smith, Amr Sabry, and Srinivasan~S. Iyengar.
\newblock Quantum computation of hydrogen bond dynamics and vibrational
  spectra.
\newblock \emph{J. Phys. Chem. Lett.}, 14:\penalty0 7256--7263, 2023.
\newblock \doi{https://doi.org/10.1021/acs.jpclett.3c01601}.

\bibitem[Dalzell et~al.(2023)Dalzell, McArdle, Berta, Bienias, Chen, Gilyén,
  Hann, Kastoryano, Khabiboulline, Kubica, Salton, Wang, and
  Brandão]{dalzell2023}
Alexander~M. Dalzell, Sam McArdle, Mario Berta, Przemyslaw Bienias, Chi-Fang
  Chen, András Gilyén, Connor~T. Hann, Michael~J. Kastoryano, Emil~T.
  Khabiboulline, Aleksander Kubica, Grant Salton, Samson Wang, and Fernando G.
  S.~L. Brandão.
\newblock Quantum algorithms: A survey of applications and end-to-end
  complexities.
\newblock \emph{arXiv:2310.03011}, 2023.
\newblock \doi{https://doi.org/10.48550/arXiv.2310.03011}.

\bibitem[MacDonell et~al.(2023)MacDonell, Navickas, Wohlers-Reichel, Valahu,
  Rao, Millican, Currington, Biercuk, Tan, Hempel, and Kassal]{D3SC02453A}
Ryan~J. MacDonell, Tomas Navickas, Tim~F. Wohlers-Reichel, Christophe~H.
  Valahu, Arjun~D. Rao, Maverick~J. Millican, Michael~A. Currington, Michael~J.
  Biercuk, Ting~Rei Tan, Cornelius Hempel, and Ivan Kassal.
\newblock Predicting molecular vibronic spectra using time-domain analog
  quantum simulation.
\newblock \emph{Chem. Sci.}, 14:\penalty0 9439--9451, 2023.
\newblock \doi{https://doi.org/10.1039/D3SC02453A}.

\bibitem[Childs et~al.(2021)Childs, Su, Tran, Wiebe, and Zhu]{childs2021theory}
Andrew~M Childs, Yuan Su, Minh~C Tran, Nathan Wiebe, and Shuchen Zhu.
\newblock Theory of trotter error with commutator scaling.
\newblock \emph{Physical Review X}, 11\penalty0 (1):\penalty0 011020, 2021.
\newblock \doi{https://doi.org/10.1103/PhysRevX.11.011020}.

\bibitem[Bravyi et~al.(2024)Bravyi, Cross, Gambetta, Maslov, Rall, and
  Yoder]{LDPC2024}
Sergey Bravyi, Andrew~W. Cross, Jay~M. Gambetta, Dmitri Maslov, Patrick Rall,
  and Theodore~J. Yoder.
\newblock High-threshold and low-overhead fault-tolerant quantum memory.
\newblock \emph{Nature}, 627:\penalty0 778--782, 2024.
\newblock \doi{https://doi.org/10.1038/s41586-024-07107-7}.

\bibitem[Hansen et~al.(2010)Hansen, Sparta, Seidler, Toffoli, and
  Christiansen]{hansen2010new}
Mikkel~B Hansen, Manuel Sparta, Peter Seidler, Daniele Toffoli, and Ove
  Christiansen.
\newblock New formulation and implementation of vibrational self-consistent
  field theory.
\newblock \emph{Journal of chemical theory and computation}, 6\penalty0
  (1):\penalty0 235--248, 2010.
\newblock \doi{https://doi.org/10.1021/ct9004454}.

\bibitem[Lee et~al.(1995)Lee, Martin, and Taylor]{Lee1995}
Timothy~J. Lee, Jan M.~L. Martin, and Peter~R. Taylor.
\newblock An accurate ab initio quartic force field and vibrational frequencies
  for ch$_4$ and isotopomers.
\newblock \emph{J. Chem. Phys.}, 102:\penalty0 254, 1995.
\newblock \doi{https://doi.org/10.1063/1.469398}.

\bibitem[Bowman et~al.(2003)Bowman, Carter, and Huang]{Bowman2003}
Joel~M. Bowman, Stuart Carter, and Xinchuan Huang.
\newblock Multimode: A code to calculate rovibrational energies of polyatomic
  molecules.
\newblock \emph{International Reviews in Physical Chemistry}, 22\penalty0
  (3):\penalty0 533--549, 2003.
\newblock \doi{https://doi.org/10.1080/0144235031000124163}.

\bibitem[Li et~al.(2001)Li, Rosenthal, and Rabitz]{Li2001}
Genyuan Li, Carey Rosenthal, and Herschel Rabitz.
\newblock High dimensional model representations.
\newblock \emph{The Journal of Physical Chemistry A}, 105\penalty0
  (33):\penalty0 7765--7777, 2001.
\newblock \doi{https://doi.org/10.1021/jp010450t}.

\bibitem[Hansen et~al.(2008)Hansen, Christiansen, Toffoli, and
  Kongsted]{Hansen2008}
Mikkel~Bo Hansen, Ove Christiansen, Daniele Toffoli, and Jacob Kongsted.
\newblock A virtual vibrational self-consistent-field method for efficient
  calculation of molecular vibrational partition functions and thermal effects
  on molecular properties.
\newblock \emph{J. Chem. Phys}, 128:\penalty0 174106, 2008.
\newblock \doi{https://doi.org/10.1063/1.2912184}.

\bibitem[Heislbetz and Rauhut(2010)]{Heislbetz2010}
Sandra Heislbetz and Guntram Rauhut.
\newblock Vibrational multiconfiguration self-consistent field theory:
  Implementation and test calculations.
\newblock \emph{J. Chem. Phys.}, 132:\penalty0 124102, 2010.
\newblock \doi{https://doi.org/10.1063/1.3364861}.

\bibitem[Császár(2012)]{csaszar2012}
Attila~G. Császár.
\newblock Anharmonic molecular force fields.
\newblock \emph{WIREs Comput Mol Sci}, 2:\penalty0 273--289, 2012.
\newblock \doi{https://doi.org/10.1002/wcms.75}.

\bibitem[Bravyi and Kitaev(2002)]{Bravyi2002}
Sergey~B. Bravyi and Alexei~Yu. Kitaev.
\newblock Fermionic quantum computation.
\newblock \emph{Annals of Physics}, 298:\penalty0 210--226, 2002.
\newblock \doi{https://doi.org/10.1006/aphy.2002.6254}.

\bibitem[Nielsen and Chuang(2011)]{mikeandike}
Michael~A. Nielsen and Isaac~L. Chuang.
\newblock \emph{Quantum Computation and Quantum Information: 10th Anniversary
  Edition}.
\newblock Cambridge University Press, USA, 10th edition, 2011.
\newblock ISBN 1107002176.
\newblock \doi{https://doi.org/10.1017/CBO9780511976667}.

\bibitem[Babbush et~al.(2018)Babbush, Gidney, Berry, Wiebe, McClean, Paler,
  Fowler, and Neven]{babbush2018}
Ryan Babbush, Craig Gidney, Dominic~W Berry, Nathan Wiebe, Jarrod McClean,
  Alexandru Paler, Austin Fowler, and Hartmut Neven.
\newblock Encoding electronic spectra in quantum circuits with linear t
  complexity.
\newblock \emph{Physical Review X}, 8\penalty0 (4):\penalty0 041015, 2018.
\newblock \doi{https://doi.org/10.1103/PhysRevX.8.041015}.

\bibitem[Dresselhaus et~al.(2000)Dresselhaus, Dresselhaus, Eklund, and
  Rao]{dresselhaus2000carbon}
Mildred~S Dresselhaus, Gene Dresselhaus, PC~Eklund, and AM~Rao.
\newblock Carbon nanotubes.
\newblock In \emph{The physics of fullerene-based and fullerene-related
  materials}, pages 331--379. Springer, 2000.
\newblock \doi{https://doi.org/10.1007/978-94-011-4038-6_9}.

\bibitem[Kinloch et~al.(2018)Kinloch, Suhr, Lou, Young, and
  Ajayan]{kinloch2018composites}
Ian~A Kinloch, Jonghwan Suhr, Jun Lou, Robert~J Young, and Pulickel~M Ajayan.
\newblock Composites with carbon nanotubes and graphene: An outlook.
\newblock \emph{Science}, 362\penalty0 (6414):\penalty0 547--553, 2018.
\newblock \doi{https://doi.org/10.1126/science.aat7439}.

\bibitem[NIS()]{NISTWebBook}
\emph{NIST Chemistry WebBook, NIST Standard Reference Database Number 69}.
\newblock \doi{https://doi.org/10.18434/T4D303}.
\newblock URL \url{https://webbook.nist.gov}.

\bibitem[Chalifoux and Tykwinski(2010)]{chalifoux2010synthesis}
Wesley~A Chalifoux and Rik~R Tykwinski.
\newblock Synthesis of polyynes to model the sp-carbon allotrope carbyne.
\newblock \emph{Nature chemistry}, 2\penalty0 (11):\penalty0 967, 2010.
\newblock \doi{https://doi.org/10.1038/nchem.828}.

\bibitem[Milani et~al.(2011)Milani, Lucotti, Russo, Tommasini, Cataldo,
  Li~Bassi, and Casari]{milani2011charge}
Alberto Milani, Andrea Lucotti, Valeria Russo, M~Tommasini, Francesco Cataldo,
  Andrea Li~Bassi, and Carlo~Spartaco Casari.
\newblock Charge transfer and vibrational structure of sp-hybridized carbon
  atomic wires probed by surface enhanced raman spectroscopy.
\newblock \emph{The Journal of Physical Chemistry C}, 115\penalty0
  (26):\penalty0 12836--12843, 2011.
\newblock \doi{https://doi.org/10.1021/jp203682c}.

\bibitem[Carter et~al.(1997)Carter, Bowman, and Harding]{CARTER19971179}
Stuart Carter, Joel~M. Bowman, and Lawrence~B. Harding.
\newblock Ab initio calculations of force fields for h2cn and c1hcn and
  vibrational energies of h2cn.
\newblock \emph{Spectrochimica Acta Part A: Molecular and Biomolecular
  Spectroscopy}, 53:\penalty0 1179--1188, 1997.
\newblock \doi{https://doi.org/10.1016/S1386-1425(97)00010-3}.

\bibitem[Rauhut(2004)]{Rauhut2004}
Guntram Rauhut.
\newblock {Efficient calculation of potential energy surfaces for the
  generation of vibrational wave functions}.
\newblock \emph{The Journal of Chemical Physics}, 121:\penalty0 9313--9322,
  2004.
\newblock \doi{https://doi.org/10.1063/1.1804174}.

\bibitem[Wong and Raman(2007)]{Sumathy2007}
Bryan~M. Wong and Sumathy Raman.
\newblock Thermodynamic calculations for molecules with asymmetric internal
  rotors—application to 1,3-butadiene.
\newblock \emph{J. Comput. Chem.}, 28:\penalty0 759–766, 2007.
\newblock \doi{https://doi.org/10.1002/jcc.20536}.

\bibitem[Trujillo and McKemmish(2021)]{Zapata2021}
Juan C.~Zapata Trujillo and Laura~K. McKemmish.
\newblock Meta-analysis of uniform scaling factors for harmonic frequency
  calculations.
\newblock \emph{Wiley Interdiscip. Rev.: Comput. Mol. Sci.}, 12:\penalty0
  e1584, 2021.
\newblock \doi{https://doi.org/10.1002/wcms.1584}.

\bibitem[Liang et~al.(2023)Liang, Feng, Liu, and Head-Gordon]{HeadGordon2023}
Jiashu Liang, Xintian Feng, Xiao Liu, and Mrtin Head-Gordon.
\newblock Analytical harmonic vibrational frequencies with
  $\textrm{VV}$10-containing density functionals: Theory, efficient
  implementation, and benchmark assessments.
\newblock \emph{J. Chem. Phys.}, 158:\penalty0 204109, 2023.
\newblock \doi{https://doi.org/10.1063/5.0152838}.

\bibitem[Childs and Wiebe(2012)]{LCU2012}
Andrew~M. Childs and Nathan Wiebe.
\newblock Hamiltonian simulation using linear combinations of unitary
  operations.
\newblock \emph{Quantum Information and Computation}, 12:\penalty0 901--924,
  2012.
\newblock \doi{https://doi.org/10.26421/QIC12.11-12-1}.

\bibitem[Chakraborty(2024)]{LCU2024}
Shantanav Chakraborty.
\newblock Implementing any linear combination of unitaries on intermediate-term
  quantum computers.
\newblock \emph{Quantum}, 8:\penalty0 1496, 2024.
\newblock \doi{https://doi.org/10.22331/q-2024-10-10-1496}.

\bibitem[Low and Chuang(2019)]{Qubitization2019}
Guang~Hao Low and Isaac~L. Chuang.
\newblock Hamiltonian simulation by qubitization.
\newblock \emph{Quantum}, 3:\penalty0 163, 2019.
\newblock \doi{https://doi.org/10.22331/q-2019-07-12-163}.

\bibitem[et~al.(2024)]{qcentricmaterial}
Yuri~Alexeev et~al.
\newblock Quantum-centric supercomputing for materials science: A perspective
  on challenges and future directions.
\newblock \emph{Future Generation Computer Systems}, 160:\penalty0 666--710,
  2024.
\newblock \doi{https://doi.org/10.1016/j.future.2024.04.060}.

\bibitem[Christiansen et~al.(version 2022.10.0)]{midas}
Ove Christiansen et~al.
\newblock Midascpp: Molecular interactions dynamics and simulation.
\newblock version 2022.10.0.
\newblock URL \url{https://midascpp.gitlab.io/pages/manual.html}.

\bibitem[Abraham et~al.(2019)]{qiskit}
H{\'e}ctor Abraham et~al.
\newblock Qiskit: An open-source framework for quantum computing.
\newblock 2019.
\newblock \doi{https://doi.org/10.5281/zenodo.2562110}.

\end{thebibliography}

%\bibliography{references}

\onecolumn\newpage
\appendix 
\section{Vibrational self-consistent field method}
\label{Appendix_A}
Vibrational self-consistent field (VSCF) is a time-independent form of self-consistent field method. The method is implemented in the MidasCpp (Molecular Interactions, Dynamics and Simulation in C++/Chemistry Program Package)~\cite{midas} and provides the reference state and optimal single-mode basis functions
(modals). For a system with $L$ modes, the VSCF ansatz for the wave function is given by
\begin{equation}
\label{eq:VSCF_State}
\ket{\Phi_{\rm VSCF}}
= e^{\kappa} \ket{\Phi_\mathbf{R}}
= = \exp \left( \sum_{l=1}^L \sum_{k_l \, h_l} \kappa_{k_l h_l} \mathcal{E}_{k_l h_l} \right)
\ket{\Phi_\mathbf{R}} \, ,
\end{equation}
with the reference state being 
\begin{equation}
\label{eq:reference_State}
\ket{\Phi_\mathbf{R}} = \prod_l^L a_{R_l}^\dagger\,  \ket{\rm{vac}}\, .
\end{equation}
The vector $\mathbf{R}$ specifies the mode that modal is occupied in the reference state and ${\op{E}}_{k_l h_l}  =  a_{k_l}^\dagger a_{h_l}$, with $a_{k_l}^\dagger $ and $a_{h_l}$ being the bosonic creation and annihilation operators, respectfully, and $a_{h_l} \ket{\rm vac} = 0$. The operator $\kappa$ is anti-hermitian and  $\exp(\kappa)$ is unitary. The exponential prefactor generates rotations among the modals. The optimal VSCF modals are obtained by implementing the variational principle on the modal rotation parameters ($\kappa_{k_l h_l}$) of the VSCF energy 
$E_{\rm VSCF} = 
 \bra{\Phi_\mathbf{R}} \exp(-\kappa) H \exp(\kappa) \ket{\Phi_\mathbf{R}}$:
\begin{equation}
\frac{\partial}{\partial \kappa_{k_l h_l}} \bra{\Phi_\mathbf{R}} \exp(-\kappa) H \exp(\kappa)  \ket{\Phi_\mathbf{R}} = 0. 
\label{eq:derivative_VSCF_Energy}
\end{equation}
By inserting the VSCF ansatz in this equation, one can obtain 
\begin{equation}
 \bra{\Phi_\mathbf{R}} \, [H, {\op{E}}_{k_l h_l} ] \,\ket{\Phi_\mathbf{R}} 
 = 0 
 \label{eq:Optimal_Criterion} 
 \end{equation}
as a criterion for having an optimal $\ket{\Phi_\mathbf{R}}$ for $\kappa=0$. 
Then one can assign an effective mean-field operator for mode $l$, \begin{equation}
 F^{l,\mathbf{R}}_{k_l h_l} =
 \bra{\Phi_\mathbf{R}} \, [[a_{k_l},H ],a_{h_l}^\dagger] 
 \,\ket{\Phi_\mathbf{R}}.
 \label{eq:EMF}
\end{equation}
Employing the second quantization commutator relations, the $F^{l,\mathbf{R}}$  
matrix elements would be directly related to the VSCF gradient terms where a zero gradient can be obtained by diagonalizing the $F^{l,\mathbf{R}}$ matrix. Then solving the VSCF equations for a given Hamiltonian means constructing the $F^{l,\mathbf{R}}$ matrix, diagonalizing it, updating the Hamiltonian parameters, calculating the VSCF energy and checking the convergence;
 see \cite{hansen2010new,csaszar2012} for more details.

The other simpler but less accurate approach here is harmonic approximation, where states $\ket{\Phi}$ are the eigenfunctions of a harmonic oscillator for each mode.
Then the one-body and two-body integrals in Eq.~(2) of the main text can be expressed as a Taylor expansion and  are easy to calculate. These integrals are implemented in Qiskit \cite{qiskit}.

%%%%%%%%%%%%%%%%%%%%%%%%%%%%%%%%%%%%%%%%%%%%%%%%%%%%

\newpage
\section{Electronic structure vs Vibrational structure }
\label{Appendix_B}
\begin{table*}[ht]
    \caption{Number of terms for different type of simulations based on the required number of qubits. The cut-off for the electronic structure simulations is 10$^{-8}$~hartree while it is 10$^{-5}$~cm$^{-1}$ for the vibrational structure ones. The numbers here are for Jordan–Wigner and  unary mappings for the electronic and vibrational structures, respectively.}
    \label{tab:sizes}
    \centering
    \begin{tabular}{cllcr}
        \hline
        Qubits & Molecule & \hspace{0.25cm} Simulation Type & Basis Functions/number of Modals & \hspace{0.5cm} Pauli Terms\\
        \hline
        10 & LiH \,  &  \textbf{electronic structure}& sto-3g &276
        \\
        12 & H$_2$O \,  & \textbf{electronic structure}& sto-3g &551
        \\
         20 & C$_2$H$_2$ or P$_1$ \,  &\textbf{electronic structure}& sto-3g &2951
        \\
         20 & LiH \, & \textbf{electronic structure}& 6-31g &5851
        \\
         24 & H$_2$O \, & \textbf{electronic structure}& 6-31g &8921
        \\
         28 & C$_2$H$_2$ or P$_1$ \, & \textbf{vibrational structure}&  4 modals &15657
        \\
        32 & N$_2$ \, & \textbf{electronic structure}& 6-31g &21521
        \\
        36 & LiH \,  & \textbf{electronic structure}& 
        cc-pvdz &63519
        \\
        40 & C$_2$H$_2$ or P$_1$ \, & \textbf{electronic structure}& 6-31g &53289
        \\
        42 & C$_2$H$_2$ or P$_1$ \, & \textbf{vibrational structure}&  6 modals  & 155279
        \\
        46 & H$_2$O \, & \textbf{electronic structure}& 
        cc-pvdz &107382
        \\
        52 & N$_2$ \, & \textbf{electronic structure}&
        cc-pvdz &145327
        \\
        52 &  P$_2$ \, & \textbf{vibrational structure}&  4 modals & 114948
        \\
        56 &  C$_2$H$_2$ or P$_1$ \, & \textbf{vibrational structure}&  8 modals & 693419
        \\
        72 &  C$_2$H$_2$ or P$_1$ \, & \textbf{electronic structure}& 
        cc-pvdz &542433
        \\
        76 &  P$_3$ \, & \textbf{vibrational structure}&  4 modals & 376974
        \\
        78 &  P$_2$ \, & \textbf{vibrational structure}&  6 modals & 1044058
        \\
        104 &  P$_2$ \, & \textbf{vibrational structure}&  8 modals & 4386132
        \\
        114 &  P$_3$ \, & \textbf{vibrational structure}&  6 modals & 3277705
        \\
        \hline
    \end{tabular}
\end{table*}

\end{document}